\documentclass[11pt,onecolumn]{IEEEtran}

\usepackage{graphicx}
\usepackage{amsmath}
\usepackage{mathrsfs}
\usepackage{mathtools}
\usepackage{amssymb}
\usepackage{float}
\usepackage{url}
\usepackage{verbatim}
\usepackage{enumerate}
\usepackage{layout}
\usepackage{bm}

\usepackage[bottom=0.8in]{geometry}

\newtheorem{theorem}{Theorem}
\newtheorem{lemma}{Lemma}
\newtheorem{proposition}{Proposition}
\newtheorem{corollary}{Corollary}

\newtheorem{definition}{Definition}
\newtheorem{example}{Example}
\newtheorem{remark}{Remark}
\newcommand{\prob}{\ensuremath{\mathbb{P}}}
\newcommand{\naturals}{\ensuremath{\mathbb{N}}}

\newcommand{\Reals}{\ensuremath{\mathbb{R}}}

\newcommand{\set}{\ensuremath{\mathcal}}

\newcommand{\es}{\varnothing}  

\newcommand{\cset}[1]{{\set{#1}}^{\textnormal{c}}} 
\newcommand{\OneTo}[1]{[#1]}
\newcommand{\FromTo}[2]{[#1:#2]}
\newcommand{\eqdef}{\triangleq} 
\newcommand{\card}[1]{|#1|}
\newcommand{\bigcard}[1]{\bigl|#1\bigr|}

\newcommand{\biggcard}[1]{\biggl|#1\biggr|}

\RequirePackage{bbm}

\newcommand{\pmfsub}[1]{{\smash{#1}\vphantom{XYZ}}}
\newcommand{\pmf}[1]{\mathsf{#1}}
\newcommand{\pmfOf}[1]{\pmf{P}_{\pmfsub{#1}}}  
\newcommand{\CondpmfOf}[2]{\pmf{P}_{\pmfsub{#1} | \pmfsub{#2}}}

\DeclareMathOperator{\Vertex}{\mathsf{V}}
\DeclareMathOperator{\Edge}{\mathsf{E}}

\newcommand{\V}[1]{\Vertex(#1)}
\newcommand{\E}[1]{\Edge(#1)}
\newcommand{\Ed}[2]{\Edge_{#1}(#2)}

\newcommand{\dH}[2]{\mathsf{d}_{\textnormal{H}} \kern-0.1em \left(#1, \kern0.1em #2\right)}
\newcommand{\bigdH}[2]{\mathsf{d}_{\textnormal{H}} \kern-0.1em \big(#1, \kern0.1em #2\big)}

\newcommand{\ch}[1]{\pmf{#1}} 
\newcommand{\chlong}[3]{\ch{#1}\colon #2 \to #3} 

\newcommand{\CondInd}[3]{#1 \perp\!\!\!\perp #2 \, | \, #3}

\DeclareMathOperator{\Entr}{H}
\DeclareMathOperator{\Info}{I}
\DeclareMathOperator{\EntrPow}{N}
\newcommand{\Hb}{\mathsf{H}_{\textnormal{b}}}

\newcommand{\Ent}[1]{\Entr(#1)}
\newcommand{\bigEnt}[1]{\Entr\bigl(#1\bigr)}

\newcommand{\BiggEnt}[1]{\Entr\Biggl(#1\Biggr)}

\newcommand{\EntPow}[1]{\EntrPow(#1)}

\newcommand{\BiggEntPow}[1]{\EntrPow\Biggl(#1\Biggr)}

\newcommand{\EntPowToPow}[2]{\EntrPow^{#1}\kern-0.1em(#2)}
\newcommand{\bigEntPowToPow}[2]{\EntrPow^{#1}\kern-0.1em\bigl(#2\bigr)}
\newcommand{\BigEntPowToPow}[2]{\EntrPow^{#1}\kern-0.1em\Bigl(#2\Bigr)}
\newcommand{\biggEntPowToPow}[2]{\EntrPow^{#1}\kern-0.1em\biggl(#2\biggr)}
\newcommand{\BiggEntPowToPow}[2]{\EntrPow^{#1}\kern-0.1em\Biggl(#2\Biggr)}

\newcommand{\EntBin}[1]{\Hb(#1)}

\newcommand{\biggEntBin}[1]{\Hb\biggl(#1\biggr)}

\newcommand{\EntCond}[2]{\Entr(#1 | \kern0.1em #2)}
\newcommand{\bigEntCond}[2]{\Entr\bigl(#1 | \kern0.1em #2\bigr)}
\newcommand{\BigEntCond}[2]{\Entr\Bigl(#1 \kern-0.1em \bigm| \kern-0.1em #2 \Bigr)}
\newcommand{\biggEntCond}[2]{\Entr\biggl(#1 \kern-0.1em \Bigm| \kern-0.1em #2 \biggr)}
\newcommand{\BiggEntCond}[2]{\Entr\Biggl(#1 \kern-0.1em \biggm| \kern-0.1em #2 \Biggr)}

\newcommand{\MInfo}[2]{\Info(#1; #2)}

\newcommand{\biggMInfo}[2]{\Info\biggl(#1; \kern0.1em #2 \biggr)}
\newcommand{\BiggMInfo}[2]{\Info\Biggl(#1; \kern0.1em #2 \Biggr)}

\newcommand{\MInfoCond}[3]{\Info(#1; #2 \kern0.1em |\kern0.1em #3)}
\newcommand{\bigMInfoCond}[3]{\Info\bigl(#1; #2 \kern0.1em |\kern0.1em #3\bigr)}
\newcommand{\BigMInfoCond}[3]{\Info\Bigl(#1; #2 \kern-0.1em \bigm| \kern-0.1em #3 \kern-0.1em \Bigr)}
\newcommand{\biggMInfoCond}[3]{\Info\biggl(#1; \kern0.1em #2 \kern-0.1em \Bigm| \kern-0.1em #3\biggr)}
\newcommand{\BiggMInfoCond}[3]{\Info\Biggl(#1; \kern0.1em #2 \kern-0.1em \biggm| \kern-0.1em #3\Biggr)}

\newcommand{\KLDivBin}[2]{\mathsf{D}_{\textnormal{b}}(#1 \kern0.1em \| \kern0.1em #2)} 
\newcommand{\bigKLDivBin}[2]{\mathsf{D}_{\textnormal{b}}\bigl(#1 \kern0.1em \| \kern0.1em #2 \bigr)}
\newcommand{\BigKLDivBin}[2]{\mathsf{D}_{\textnormal{b}}\Bigl(#1 \kern0.1em \| \kern0.1em #2 \Bigr)}
\newcommand{\biggKLDivBin}[2]{\mathsf{D}_{\textnormal{b}}\biggl(#1 \kern0.1em \| \kern0.1em #2 \biggr)}
\newcommand{\BiggKLDivBin}[2]{\mathsf{D}_{\textnormal{b}}\Biggl(#1 \kern0.1em \| \kern0.1em #2 \Biggr)}
\newcommand{\bigKLDiv}[2]{\mathsf{D}\bigl(#1 \kern0.1em \| \kern0.1em #2 \bigr)}
\newcommand{\BigKLDiv}[2]{\mathsf{D}\Bigl(#1 \kern-0.2em \bigm\| \kern-0.2em #2 \Bigr)}
\newcommand{\biggKLDiv}[2]{\mathsf{D}\biggl(#1 \kern-0.2em \Bigm\| \kern-0.2em #2 \biggr)}
\newcommand{\BiggKLDiv}[2]{\mathsf{D}\Biggl(#1 \kern-0.1em \biggm\| \kern-0.1em #2 \Biggr)}

\newcommand{\I}[1]{{\mathbbm 1}\{#1\}}
\newcommand{\bigI}[1]{{\mathbbm 1}\bigl\{#1 \bigr\}}
\newcommand{\BigI}[1]{{\mathbbm 1}\Bigl\{#1 \Bigr\}}

\newcommand{\BigPrv}[1]{\Pr\Bigl[#1\Bigr]}

\newcommand{\Prs}[1]{\Pr(#1)} 

\begin{document}
\thispagestyle{empty}
\setcounter{page}{1}
\setlength{\baselineskip}{1.15\baselineskip}

\title{Information Inequalities via Submodularity and a Problem in Extremal Graph Theory\\[0.2cm]}
\author{Igal Sason
\thanks{
I. Sason is with the Andrew \& Erna Viterbi Faculty of Electrical and Computer Engineering,
and with the Department of Mathematics. Both affiliations are at the Technion - Israel Institute of Technology,
Haifa 3200003, Israel (e-mail: eeigal@technion.ac.il).}}

\maketitle
\thispagestyle{empty}

\begin{abstract}
The present paper offers, in its first part, a unified approach for the derivation of
families of inequalities for set functions which satisfy sub/supermodularity properties. It applies
this approach for the derivation of information inequalities with Shannon information measures.
Connections of the considered approach to a generalized version of Shearer's lemma, and other related
results in the literature are considered. Some of the derived information inequalities are new, and
also known results (such as a generalized version of Han's inequality) are reproduced in a simple and
unified way. In its second part, this paper applies the generalized Han's inequality to analyze a problem
in extremal graph theory. This problem is motivated and analyzed from the perspective of information theory,
and the analysis leads to generalized and refined bounds.
The two parts of this paper are meant to be independently accessible to the reader.
\end{abstract}


\section{Introduction}
\label{section: introduction}

Information measures and information inequalities are of fundamental importance and
wide applicability in the study of feasibility and infeasibility results in information theory,
while also offering very useful tools which serve to deal with
interesting problems in various fields of mathematics \cite{CoverT06, DemboCT91}. The characterization
of information inequalities has been of interest for decades (see, e.g., \cite{Chan11}, \cite{MartinPY16} and references
therein), mainly triggered by their indispensable role in proving direct and converse results for
channel coding and data compression for single and multi-user information systems.
Information inequalities, which apply to classical and generalized information measures, have also demonstrated
far-reaching consequences beyond the study of the coding theorems and fundamental limits of communication systems.
One such remarkable example (among many) is the usefulness of information measures
and information inequalities in providing information-theoretic proofs in the field of combinatorics and graph
theory (see, e.g., \cite{BabuR14}--\cite{MadimanT_IT10}).

A basic property that is commonly used for the characterization of information inequalities relies on the nonnegativity of
the (conditional and unconditional) Shannon entropy of discrete random variables, the nonnegativity of the
(conditional and unconditional) relative entropy and the Shannon mutual information of general random variables,
and the chain rules which hold for these classical information measures. A byproduct of these properties is the sub/supermodularity
of some classical information measures, which also prove to be useful by taking advantage of the vast literature
on sub/supermodular functions and polymatroids \cite{MadimanT_IT10}--\cite{KIshiOY_ISIT2014}.
Another instrumental information inequality is the entropy power inequality, which dates back to Shannon \cite{Shannon1948}.
It has been extensively generalized for different types of random variables and generalized entropies,
studied in regard to its geometrical relations \cite{MadimanMX17}, and it has been also ubiquitously used for the
analysis of various information-theoretic problems.

Among the most useful information inequalities are Han's inequality \cite{Han78}, its generalized versions
(e.g., \cite{MadimanMT_ITW10, Fujishige78, Tian_ISIT11, KIshiOY_ISIT2014}),
and Shearer's lemma \cite{ChungGFS86} with its generalizations and refinements (e.g., \cite{MadimanMT_ITW10,KIshiOY_ISIT2014,PolyanskiyW19}).
In spite of their simplicity, these inequalities prove to be useful in information theory, and other diverse fields of mathematics and engineering
(see, e.g., \cite{Boucheron_Lugosi_Massart_book, PolyanskiyW19}).
More specifically in regard to these inequalities, in Proposition~1 of \cite{MadimanT_IT10}, Madiman and Tetali introduced an information inequality which
can be specialized to Han's inequality, and which also refines Shearer's lemma while also providing a
counterpart result. In \cite{Tian_ISIT11}, Tian generalized Han's inequality by relying on the
sub/supermodularity of the unconditional/conditional Shannon entropy.
Likewise, the work in \cite{KIshiOY_ISIT2014} by Kishi {\em et al.} relies on the sub/supermodularity properties of
Shannon information measures, and it provides refinements of Shearer's lemma and Han's inequality.
Apart of the refinements of these classical and widely-used inequalities in \cite{KIshiOY_ISIT2014},
the suggested approach in the present work can be viewed in a sense as a (nontrivial) generalization and extension
of a result in \cite{KIshiOY_ISIT2014} (to be explained in Section~\ref{subsection: Connections to Literature}).

This work is focused on the derivation of information inequalities via submodularity and nonnegativity properties, and on
a problem in extremal graph theory whose analysis relies on an information inequality.
The field of extremal graph theory, which is a subfield of extremal combinatorics, was among the early and
fast developing branches of graph theory during the 20th century. Extremal graph theory explores
the relations between properties of graphs such as its order, size, chromatic number or maximal and minimal
degrees, under some constraints on the graph (by, e.g., considering graphs of a fixed order, and by also
imposing a type of a forbidden subgraph).
The interested reader is referred to the comprehensive textbooks \cite{Jukna} and \cite{Bollobas}
on the vast field of extremal combinatorics and extremal graph theory. For the
purpose of the suggested problem and analysis, the presentation here is however self-contained.

This paper suggests an approach for the derivation of families of inequalities
for set functions, and it applies it to obtain information inequalities with
Shannon information measures that satisfy sub/supermodularity and monotonicity properties.
Some of the derived information inequalities are new, while some known results (such as the
generalized version of Han's inequality \cite{Fujishige78}) are reproduced as corollaries
in a simple and unified way. This paper also applies the generalized Han's inequality to
analyze a problem in extremal graph theory, with an information-theoretic proof and interpretation.
The analysis leads to some generalized and refined bounds in comparison to
the insightful results in Theorems~4.2 and 4.3 of \cite{Boucheron_Lugosi_Massart_book}.

The paper is structured as follows:
Section~\ref{section: preliminaries} provides essential notation and preliminary material for this paper.
Section~\ref{section: methodologies} presents a new methodology for the derivation of families of
inequalities for set functions which satisfy sub/supermodularity properties (Theorem~\ref{theorem: Methodology 1}).
The suggested methodology is then applied in Section~\ref{section: methodologies} for the derivation of
information inequalities by relying on sub/supermodularity properties of Shannon information measures.
Section~\ref{section: methodologies} also considers connections of the suggested approach to a
generalized version of Shearer's lemma, and to other results in the literature. Most of the results in
Section~\ref{section: methodologies} are proved in Section~\ref{section: proofs}.
Section~\ref{section: problem} applies the generalized Han's inequality to a problem in extremal
graph theory (Theorem~\ref{theorem: graphs}). A byproduct of Theorem~\ref{theorem: graphs},
which is of interest in its own right, is also analyzed in Section~\ref{section: problem}
(Theorem~\ref{theorem: influence}). The presentation and analysis in Section~\ref{section: problem}
is accessible to the reader, independently of the earlier material on information inequalities in
Sections \ref{section: methodologies} and \ref{section: proofs}. Some additional proofs, mostly
for making the paper self-contained or for suggesting an alternative proof, are relegated to the
appendices (Appendices~\ref{appendix A: proof} and~\ref{Appendix: Proof - Shearer - submodularity}).

\section{Preliminaries and Notation}
\label{section: preliminaries}

The present section provides essential notation and preliminary material for this paper.
\begin{itemize}
\item $\naturals \eqdef \{1, 2, \ldots \}$ denotes the set of natural numbers.
\item $\Reals$ denotes the set of real numbers, and $\Reals_+$ denotes the set
of nonnegative real numbers.
\item $\es$ denotes the empty set.
\item $2^\Omega \eqdef \bigl\{\set{A}: \set{A} \subseteq \Omega\bigr\}$ denotes the power
set of a set $\Omega$ (i.e., the set of all subsets of $\Omega$).
\item $\cset{T} \eqdef \Omega \setminus \set{T}$ denotes the complement of
a subset $\set{T}$ in $\Omega$.
\item $\I{E}$ is an indicator of $E$; it is~1
if event $E$ is satisfied, and zero otherwise.
\item $\OneTo{n} \eqdef \{1, \ldots, n\}$ for every $n \in \naturals$;
\item $X^n \eqdef (X_1, \ldots, X_n)$ denotes an $n$-dimensional random vector;
\item $X_{\set{S}} \eqdef (X_i)_{i \in \set{S}}$ is a random vector for
a nonempty subset $\set{S} \subseteq \OneTo{n}$; if $\set{S} = \es$,
then $X_{\set{S}}$ is an empty set, and conditioning on $X_{\set{S}}$ is void.
\item Let $X$ be a discrete random variable that takes its values on a set $\set{X}$,
and let $\pmfOf{X}$ be the probability mass function (PMF) of $X$. The
{\em Shannon entropy} of $X$ is given by
\begin{eqnarray}
\label{eq: entropy}
\Ent{X} \eqdef -\sum_{x \in \set{X}} \pmfOf{X}(x) \, \log \pmfOf{X}(x),
\end{eqnarray}
where throughout this paper, we take all logarithms to base~2.
\item The {\em binary entropy function} $\Hb \colon [0,1] \to [0, \log 2]$ is
given by
\begin{eqnarray}
\label{eq2: EntBin}
\EntBin{p} \eqdef -p \log p -(1-p) \log(1-p),  \quad p \in [0,1],
\end{eqnarray}
where, by continuous extension, the convention $0 \log 0 = 0$ is used.
\item Let $X$ and $Y$ be discrete random variables with a joint PMF $\pmfOf{XY}$, and a conditional
PMF of $X$ given $Y$ denoted by $\CondpmfOf{X}{Y}$. The {\em conditional entropy} of $X$ given $Y$
is defined as

\vspace*{-0.3cm}
\begin{subequations} \label{eq: conditional entropy}
\begin{align}
\label{eq: conditional entropy 1}
\EntCond{X}{Y} & \eqdef -\sum_{(x,y) \in \set{X} \times \set{Y}}
\pmfOf{XY}(x,y) \log \CondpmfOf{X}{Y}(x|y) \\[0.1cm]
\label{eq: conditional entropy 1.5}
&= \sum_{y \in \set{Y}} \pmfOf{Y}(y) \, \EntCond{X}{Y=y},
\end{align}
\end{subequations}

\vspace*{-0.3cm}
and
\begin{align}
\EntCond{X}{Y} &= \Ent{X,Y} - \Ent{Y}.  \label{eq: conditional entropy 2}
\end{align}
\item The {\em mutual information} between $X$ and $Y$ is symmetric in $X$ and $Y$,
and it is given by

\vspace*{-0.3cm}
\begin{subequations} \label{eq: MI}
\begin{align}
\MInfo{X}{Y} &= \Ent{X} + \Ent{Y} - \Ent{X,Y} \label{eq1: MI} \\
&= \Ent{X} - \EntCond{X}{Y}  \label{eq2: MI} \\
&= \Ent{Y} - \EntCond{Y}{X}. \label{eq3: MI}
\end{align}
\end{subequations}
\item The {\em conditional mutual information} between two random variables
$X$ and $Y$, given a third random variable $Z$, is symmetric in $X$ and $Y$
and it is given by

\vspace*{-0.3cm}
\begin{subequations}  \label{eq: Cond MI}
\begin{align}
\MInfoCond{X}{Y}{Z} &= \EntCond{X}{Z} - \EntCond{X}{Y,Z} \label{eq1: Cond MI} \\
&= \Ent{X,Z} + \Ent{Y,Z} - \Ent{Z} - \Ent{X,Y,Z}. \label{eq3: Cond MI}
\end{align}
\end{subequations}
\item For continuous random variables, the sums in \eqref{eq: entropy}
and \eqref{eq: conditional entropy} are replaced with integrals, and
the PMFs are replaced with probability densities.
The entropy of a continuous random variable is named {\em differential
entropy}.
\item
For an $n$-dimensional random vector $X^n$, the {\em entropy power} of $X^n$ is given by
\begin{eqnarray} \label{eq: entropy power}
\EntPow{X^n} \eqdef \exp\Bigl( \tfrac{2}{n} \; \Ent{X^n} \Bigr),
\end{eqnarray}
where the base of the exponent is identical to the base of the logarithm in \eqref{eq: entropy}.
\end{itemize}
We rely on the following basic properties of the Shannon information measures:
\begin{itemize}
\item Conditioning cannot increase the entropy, i.e.,
\begin{eqnarray}
\label{eq: conditioning reduces entropy}
\EntCond{X}{Y} \leq \Ent{X},
\end{eqnarray}
with equality in \eqref{eq: conditioning reduces entropy} if and only if $X$ and $Y$
are independent.
\item Generalizing \eqref{eq: conditional entropy 2}
to $n$-dimensional random vectors gives the chain rule
\begin{eqnarray}
\label{eq: chain rule}
\Ent{X^n} = \sum_{i=1}^n \EntCond{X_i}{X^{i-1}}.
\end{eqnarray}
\item The {\em subadditivity property of the entropy} is implied by
\eqref{eq: conditioning reduces entropy} and \eqref{eq: chain rule}:
\begin{eqnarray}
\label{eq: subadditivity Ent}
\Ent{X^n} \leq \sum_{i=1}^n \Ent{X_i},
\end{eqnarray}
with equality in \eqref{eq: subadditivity Ent} if and only if $X_1, \ldots, X_n$ are
independent random variables.
\item {\em Nonnegativity of the (conditional) mutual information}:
In light of \eqref{eq: MI} and \eqref{eq: conditioning reduces entropy},
$\MInfo{X}{Y} \geq 0$ with equality if and only if $X$ and $Y$ are independent.
More generally, $\MInfoCond{X}{Y}{Z} \geq 0$ with equality if and only if
$X$ and $Y$ are conditionally independent given $Z$.
\end{itemize}

Let $\Omega$ be a finite and non-empty set, and let
$f \colon 2^{\Omega} \to \Reals$ be a real-valued
set function (i.e., $f$ is defined for all subsets of $\Omega$).
The following definitions are used.

\begin{definition}[Sub/Supermodular function]
\label{definition: submodularity}
The set function $f \colon 2^{\Omega} \to \Reals$ is {\em submodular} if
\begin{eqnarray} \label{eq: submodularity}
f(\set{T}) + f(\set{S}) \geq f(\set{T} \cup \set{S}) + f(\set{T} \cap \set{S}),
\qquad \forall \; \set{S}, \set{T} \subseteq \Omega
\end{eqnarray}
Likewise, $f$ is {\em supermodular} if $-f$ is submodular.
\end{definition}

An identical characterization of submodularity
is the diminishing return property (see, e.g.,
Proposition~2.2 in \cite{Bach13}), where
a set function $f \colon 2^{\Omega} \to \Reals$ is submodular if and only if
\begin{eqnarray} \label{DMR}
\set{S} \subset \set{T} \subset \Omega, \; \; \omega \in \cset{T}
\; \Longrightarrow \;  f(\set{S} \cup \{\omega\}) - f(\set{S})
\geq  f(\set{T} \cup \{\omega\}) - f(\set{T}).
\end{eqnarray}
This means that the larger is the set, the smaller is the
increase in $f$ when a new element is added.

\begin{definition}[Monotonic function]
\label{definition: monotonicity}
The set function $f \colon 2^{\Omega} \to \Reals$ is {\em monotonically
increasing} if
\begin{eqnarray} \label{eq: increasing}
\set{S} \subseteq \set{T} \subseteq \Omega \; \Longrightarrow \; f(\set{S}) \leq f(\set{T}).
\end{eqnarray}
Likewise, $f$ is {\em monotonically decreasing} if $-f$ is monotonically
increasing.
\end{definition}

\begin{definition}[Polymatroid, ground set and rank function]
\label{definition: polymatroid}
Let $f \colon 2^{\Omega} \to \Reals$ be submodular
and monotonically increasing set function with $f(\es) = 0$.
The pair $(\Omega, f)$ is called a {\em polymatroid},
$\Omega$ is called a {\em ground set}, and $f$ is called a {\em rank function}.
\end{definition}

\begin{definition}[Subadditive function]
\label{definition: subadditivity}
The set function $f \colon 2^{\Omega} \to \Reals$ is {\em subadditive} if,
for all $\set{S}, \set{T} \subseteq \Omega$,
\begin{eqnarray} \label{eq: subadditive function}
f(\set{S} \cup \set{T}) \leq f(\set{S}) + f(\set{T}).
\end{eqnarray}
\end{definition}
A nonnegative and submodular set function is subadditive (this readily follows from \eqref{eq: submodularity}
and \eqref{eq: subadditive function}).
The next proposition introduces results from \cite{Fujishige78}, \cite{Krause_UAI05} and \cite{Madiman_ITW08}.
For the sake of completeness, we provide a proof in Appendix~\ref{appendix A: proof}.

\begin{proposition} \label{proposition: information measures, polymatroids}
Let $\Omega$ be a finite and non-empty set, and let $\{X_\omega\}_{\omega \in \Omega}$ be a collection of
discrete random variables. Then, the following holds:
\begin{enumerate}[a)]
\item The set function $f \colon 2^{\Omega} \to \Reals$, given by
\begin{eqnarray}   \label{entropic function}
f(\set{T}) \eqdef \Ent{X_{\set{T}}}, \quad \set{T} \subseteq \Omega,
\end{eqnarray}
is a rank function.
\item The set function $f \colon 2^{\Omega} \to \Reals$, given by
\begin{eqnarray}   \label{set function 2}
f(\set{T}) \eqdef \EntCond{X_{\set{T}}}{X_{\cset{T}}}, \quad  \set{T} \subseteq \Omega,
\end{eqnarray}
is supermodular, monotonically increasing, and $f(\es) = 0$.
\item The set function $f \colon 2^{\Omega} \to \Reals$, given by
\begin{eqnarray}   \label{set function 3}
f(\set{T}) \eqdef \MInfo{X_{\set{T}}}{X_{\cset{T}}}, \quad  \set{T} \subseteq \Omega,
\end{eqnarray}
is submodular, $f(\es) = 0$, but $f$ is not a rank function. The latter holds since the
equality $f(\set{T}) = f(\cset{T})$, for all $\set{T} \subseteq \Omega$, implies that
$f$ is not a monotonic function.
\item Let $\set{U}, \set{V} \subseteq \Omega$ be disjoint subsets, and let
the entries of the random vector $X_{\set{V}}$ be conditionally independent given
$X_{\set{U}}$. Then, the set function $f \colon 2^{\set{V}} \to \Reals$ given by
\begin{eqnarray}   \label{set function 4}
f(\set{T}) \eqdef \MInfo{X_{\set{U}}}{X_{\set{T}}}, \quad \set{T} \subseteq \set{V},
\end{eqnarray}
is a rank function.
\item Let $X_{\Omega} = \{X_\omega\}_{\omega \in \Omega}$
be independent random variables, and let the set function $f \colon 2^{\Omega} \to \Reals$ be given by
\begin{eqnarray}   \label{set function 5}
f(\set{T}) \eqdef \BiggEnt{ \, \sum_{\omega \in \set{T}} X_{\omega}}, \quad \set{T} \subseteq \Omega.
\end{eqnarray}
Then, $f$ is a rank function.
\end{enumerate}
\end{proposition}

The following proposition addresses the setting of general alphabets.
\begin{proposition} \label{proposition: information measures, continuous RVs}
For general alphabets, the set functions $f$ in
\eqref{entropic function} and \eqref{set function 3}--\eqref{set function 5}
are submodular, and the set function $f$ in \eqref{set function 2} is supermodular
with $f(\es) \eqdef 0$. Moreover, the function in
\eqref{set function 4} stays to be a rank function, and the function
in \eqref{set function 5} stays to be monotonically increasing.
\end{proposition}
\begin{proof}
The sub/supermodularity properties in
Proposition~\ref{proposition: information measures, polymatroids} are
preserved due to the nonnegativity of the (conditional) mutual
information. The monotonicity property
of the functions in \eqref{set function 4} and \eqref{set function 5}
is preserved also in the general alphabet setting due to
\eqref{eq3f: monotonicity - set function 4} and \eqref{eq1c: monotonicity - set function 5},
and the mutual information in \eqref{set function 4} is nonnegative.
\end{proof}

\begin{remark}
In contrast to the entropy of discrete random variables, the differential
entropy of continuous random variables is {\em not} functionally submodular
in the sense of Lemma~A.2 in \cite{Tao10}. This refers to a different form
of submodularity, which was needed by Tao \cite{Tao10} to prove sumset
inequalities for the entropy of discrete random variables. A follow-up
study in \cite{KM14} by Kontoyiannis and Madiman required substantially
new proof strategies for the derivation of sumset inequalities with the
differential entropy of continuous random variables. The basic property
which replaces the discrete functional submodularity is the data-processing
property of mutual information \cite{KM14}. In the context of the present
work, where the commonly used definition of submodularity is used (see
Definition~\ref{definition: submodularity}), the Shannon entropy of discrete
random variables and the differential entropy of continuous random variables
are both submodular set functions.
\end{remark}

\bigskip
We rely, in this paper, on the following standard terminology for graphs.
An undirected graph $G$ is an ordered pair $G = (V,E)$
where {$V = \V{G}$} is a set of elements, and {$E = \E{G}$} is a set of
2-element subsets (pairs) of $V$. The elements of $V$ are called the
vertices of $G$, and the elements of $E$ are called the edges of $G$. We use the
notation $V = \V{G}$ and $E = \E{G}$ for the sets of vertices and edges, respectively,
in the graph $G$. The number of vertices in a finite graph $G$ is called the order of
$G$, and the number of edges is called the size of $G$. Throughout this paper, we
assume that the graph $G$ is undirected and finite; it is also assumed to be a simple graph,
i.e., it has no loops (no edge connects a vertex in $G$ to itself) and there are
no multiple edges which connect a pair of vertices in $G$. If $e = \{u,v\} \in \E{G}$, then
the vertices $u$ and $v$ are the two ends of the edge $e$. The elements $u$ and $v$ are
adjacent vertices (neighbors) if they are connected by an edge in $G$, i.e., if
$e = \{u, v\} \in \E{G}$.

\section{Inequalities via Submodularity}
\label{section: methodologies}

\subsection{A New Methodology}
\label{subsection: a new methodology}

The present subsection presents a new methodology for the derivation of families of
inequalities for set functions, and in particular inequalities with information measures.
The suggested methodology relies, to large extent, on the notion of submodularity
of set functions, and it is presented in the next theorem.

\begin{theorem}
\label{theorem: Methodology 1}
Let $\Omega$ be a finite set with $\card{\Omega} = n$.
Let $f \colon 2^{\Omega} \to \Reals$ with
$f(\es) = 0$, and $g \colon \Reals \to \Reals$. Let
the sequence $\bigl\{t_k^{(n)}\bigr\}_{k=1}^n$ be given by
\begin{eqnarray}
\label{01.03.2021b1}
t_k^{(n)} \eqdef \frac1{\binom{n}{k}} \sum_{\set{T} \subseteq \Omega: \, \card{\set{T}}=k}
g\biggl(\frac{f(\set{T})}{k} \biggr), \qquad k \in \OneTo{n}.
\end{eqnarray}
\begin{enumerate}[a)]
\item If $f$ is submodular, and $g$ is monotonically increasing and convex, then
the sequence $\bigl\{t_k^{(n)}\bigr\}_{k=1}^n$ is monotonically decreasing, i.e.,
\begin{eqnarray}
\label{01.03.2021b2}
t_1^{(n)} \geq t_2^{(n)} \geq \ldots \geq t_n^{(n)} = g\biggl(\frac{f(\Omega)}{n} \biggr).
\end{eqnarray}
In particular,
\begin{eqnarray}
\label{04.03.2021a1}
\sum_{\set{T} \subseteq \Omega: \, \card{\set{T}}=k}
g\biggl(\frac{f(\set{T})}{k} \biggr) \geq \binom{n}{k} \,
g\biggl(\frac{f(\Omega)}{n} \biggr), \qquad k \in \OneTo{n}.
\end{eqnarray}
\item If $f$ is submodular, and $g$ is monotonically decreasing and concave, then
the sequence $\bigl\{t_k^{(n)}\bigr\}_{k=1}^n$ is monotonically increasing.
\item If $f$ is supermodular, and $g$ is monotonically increasing and concave, then
the sequence $\bigl\{t_k^{(n)}\bigr\}_{k=1}^n$ is monotonically increasing.
\item If $f$ is supermodular, and $g$ is monotonically decreasing and convex, then
the sequence $\bigl\{t_k^{(n)}\bigr\}_{k=1}^n$ is monotonically decreasing.
\end{enumerate}
\end{theorem}
\begin{proof}
See Section~\ref{subsection: Proof of Methodology 1}.
\end{proof}

\begin{corollary}
\label{corollary: 05.03.21a}
Let $\Omega$ be a finite set with $\card{\Omega} = n$, $f \colon 2^{\Omega} \to \Reals$,
and $g \colon \Reals \to \Reals$ be convex and monotonically increasing.
If
\begin{itemize}
\item $f$ is a rank function,  \vspace*{0.15cm}

\item $g(0)>0$ or there is $\ell \in \naturals$ such that
$g(0) = \ldots = g^{(\ell-1)}(0) = 0$ with $g^{(\ell)}(0) > 0$,  \vspace*{0.15cm}

\item $\{k_n\}_{n=1}^{\infty}$ is a sequence such that
$k_n \in \OneTo{n}$ for all $n \in \naturals$ with
$k_n \underset{n \to \infty} \longrightarrow \infty$,
\end{itemize}
then
\begin{eqnarray} \label{19.04.22a2}
\lim_{n \to \infty} \Biggl\{ \frac1n \, \log \Biggl( \sum_{\set{T} \subseteq \Omega:
\, \card{\set{T}}=k_n} g\biggl(\frac{f(\set{T})}{k_n} \biggr) \Biggr) - \biggEntBin{\frac{k_n}{n}} \Biggr\} = 0.
\end{eqnarray}
Furthermore, if $\underset{n \to \infty}{\lim} \frac{k_n}{n} = \beta \in [0,1]$, then
\begin{eqnarray} \label{04.03.2021a2}
\lim_{n \to \infty} \frac1n \, \log \Biggl( \sum_{\set{T} \subseteq \Omega:
\, \card{\set{T}}=k_n} g\biggl(\frac{f(\set{T})}{k_n} \biggr) \Biggr) = \EntBin{\beta}.
\end{eqnarray}
\end{corollary}
\begin{proof}
See Section~\ref{subsection: proof of corollary 1}.
\end{proof}

\vspace*{0.1cm}
\begin{corollary}
\label{corollary: 05.03.21b}
Let $\Omega$ be a finite set with $\card{\Omega} = n$, and $f \colon 2^{\Omega} \to \Reals$ be
submodular and nonnegative with $f(\es) = 0$. Then,
\begin{enumerate}[a)]
\item
For $\alpha \geq 1$ and $k \in \OneTo{n-1}$
\begin{eqnarray}
\label{01.03.2021b3}
\sum_{\set{T} \subseteq \Omega: \, \card{\set{T}}=k}  \bigl( f^\alpha(\Omega) - f^\alpha(\set{T})
\bigr) \leq c_{\alpha}(n,k) \, f^\alpha(\Omega),
\end{eqnarray}
with
\begin{eqnarray}
\label{eq: c}
c_{\alpha}(n,k) \eqdef \biggl(1-\frac{k^\alpha}{n^\alpha}\biggr) \, \binom{n}{k}.
\end{eqnarray}
For $\alpha=1$, \eqref{01.03.2021b3} holds with $c_1(n,k)=\binom{n-1}{k}$ regardless
of the nonnegativity of~$f$.
\item If $f$ is also monotonically increasing (i.e., $f$ is a rank function), then for $\alpha \geq 1$
\begin{eqnarray}
\label{01.03.2021b4}
\Bigl(\frac{k}{n}\Bigr)^{\alpha-1} \binom{n-1}{k-1} \, f^\alpha(\Omega) \leq
\sum_{\set{T} \subseteq \Omega: \, \card{\set{T}}=k}
f^\alpha(\set{T}) \leq \binom{n}{k} \, f^\alpha(\Omega), \qquad k \in \OneTo{n}.
\end{eqnarray}
\end{enumerate}
\end{corollary}

\begin{proof}
See Section~\ref{subsection: proof of corollary 2}.
\end{proof}

\vspace*{0.1cm}
Corollary~\ref{corollary: 05.03.21b} is next specialized to
reproduce Han's inequality \cite{Han78}, and a generalized
version of Han's inequality \cite[Section~4]{Fujishige78}.

Let $X^n = (X_1, \ldots, X_n)$ be a random vector with
finite entropies $\Ent{X_i}$ for all $i \in \OneTo{n}$. The set function
$f \colon 2^{\OneTo{n}} \to [0, \infty)$, given by $f(\set{T}) = \Ent{X_{\set{T}}}$
for all $\set{T} \subseteq \OneTo{n}$, is submodular \cite{Fujishige78}
(see Proposition~\ref{proposition: information measures, polymatroids}{a}
and Proposition~\ref{proposition: information measures, continuous RVs}).
From \eqref{01.03.2021b3}, the following holds:
\begin{enumerate}[a)]
\item Setting $\alpha=1$ in \eqref{01.03.2021b3} implies that,
for all $k \in \OneTo{n-1}$,

\vspace*{-0.3cm}
\begin{subequations} \label{eq1: Han's inequality}
\begin{align}
\sum_{1 \leq i_1 < \ldots < i_k \leq n} \bigl( \Ent{X^n} - \Ent{X_{i_1}, \ldots, X_{i_k}} \bigr)
&\leq \biggl(1 - \frac{k}{n}\biggr) \, \binom{n}{k} \, \Ent{X^n} \\
&= \binom{n-1}{k} \, \Ent{X^n},
\end{align}
\end{subequations}
\item Consequently, setting $k=n-1$ in \eqref{eq1: Han's inequality} gives
\begin{eqnarray} \label{eq2: Han's inequality}
\sum_{i=1}^n \bigl( \Ent{X^n} - \Ent{X_1, \ldots, X_{i-1}, X_{i+1}, \ldots, X_n} \bigr)
\leq \Ent{X^n},
\end{eqnarray}
which gives Han's inequality.
\end{enumerate}

Further applications of Theorem~\ref{theorem: Methodology 1} lead to the next corollary,
which partially introduces some known results that have been proved on a case-by-case
basis in \cite[Theorems~17.6.1--17.6.3]{CoverT06} and \cite[Section~2]{DemboCT91}.
In particular, the monotonicity properties of the sequences in \eqref{03.03.2021b3},
\eqref{eq1: 14.03.2021a}, \eqref{eq1: 15.03.2021a} and \eqref{03.03.2021b4} were
proved in Theorems~1 and 2, and Corollaries~1 and 2 of \cite{DemboCT91}.
Both known and new results are readily obtained here, in a unified way, from
Theorem~\ref{theorem: Methodology 1}. The utility of one of these inequalities
in extremal combinatorics is discussed in the continuation to this subsection
(see Proposition~\ref{proposition: projections}), providing a natural generalization
of a beautiful combinatorial result in \cite[Section~3.2]{Radhakrishnan01}.

\vspace*{0.1cm}
\begin{corollary}  \label{corollary: monotonicity - sums}
Let $\{X_i\}_{i=1}^n$ be random variables with finite entropies. Then,
the following holds:
\begin{enumerate}[a)]
\item The sequences
\begin{eqnarray}
\label{03.03.2021b3}
&& h_k^{(n)} \eqdef \frac1{\binom{n}{k}} \sum_{\set{T} \subseteq \OneTo{n}: \, \card{\set{T}}=k}
\frac{\Ent{X_{\set{T}}}}{k}, \qquad k \in \OneTo{n},  \\[0.1cm]
\label{eq2: 15.03.2021a}
&& \ell_k^{(n)} \eqdef \frac1{\binom{n}{k}} \sum_{\set{T} \subseteq \OneTo{n}: \, \card{\set{T}} = k}
\frac{\MInfo{X_{\set{T}}}{X_{\cset{T}}}}{k}, \quad k \in \OneTo{n}
\end{eqnarray}
are monotonically decreasing in $k$. If $\{X_i\}_{i=1}^n$ are independent, then also the sequence
\begin{eqnarray}  \label{eq1: 14.03.2021a}
m_k^{(n)} \eqdef \frac1{\binom{n-1}{k-1}} \sum_{\set{T} \subseteq \OneTo{n}: \, \card{\set{T}} = k}
\BiggEnt{\, \sum_{\omega \in \set{T}} X_\omega}, \quad k \in \OneTo{n}
\end{eqnarray}
is monotonically decreasing in $k$.
\item The sequence
\begin{eqnarray}  \label{eq1: 15.03.2021a}
r_k^{(n)} \eqdef \frac1{\binom{n}{k}} \sum_{\set{T} \subseteq \OneTo{n}: \, \card{\set{T}} = k}
\frac{\EntCond{X_{\set{T}}}{X_{\cset{T}}}}{k}, \quad k \in \OneTo{n}
\end{eqnarray}
is monotonically increasing in $k$.
\item For every $r>0$, the sequences
\begin{eqnarray}
\label{03.03.2021b4}
&& s_k^{(n)}(r) \eqdef \frac1{\binom{n}{k}} \sum_{\set{T} \subseteq \OneTo{n}: \, \card{\set{T}}=k}
\EntPowToPow{r}{X_{\set{T}}}, \qquad k \in \OneTo{n},  \\[0.1cm]
\label{eq3: 15.03.2021a}
&& u_k^{(n)}(r) \eqdef \frac1{\binom{n}{k}} \sum_{\set{T} \subseteq \OneTo{n}: \, \card{\set{T}} = k}
\exp\Biggl(-\frac{r \, \EntCond{X_{\set{T}}}{X_{\cset{T}}}}{k} \Biggr), \quad k \in \OneTo{n}, \\[0.1cm]
\label{eq4: 15.03.2021a}
&& v_k^{(n)}(r) \eqdef \frac1{\binom{n}{k}} \sum_{\set{T} \subseteq \OneTo{n}: \, \card{\set{T}} = k}
\exp\Biggl(\frac{r \, \MInfo{X_{\set{T}}}{X_{\cset{T}}}}{k} \Biggr), \quad k \in \OneTo{n}
\end{eqnarray}
are monotonically decreasing in $k$. If $\{X_i\}_{i=1}^n$ are independent, then also the sequence
\begin{eqnarray}
\label{eq2: 14.03.2021a}
w_k^{(n)}(r) \eqdef \frac1{\binom{n}{k}} \sum_{\set{T} \subseteq \OneTo{n}: \, \card{\set{T}} = k}
\BiggEntPowToPow{r}{\, \sum_{\omega \in \set{T}} X_\omega}, \quad k \in \OneTo{n}
\end{eqnarray}
is monotonically decreasing in $k$.
\end{enumerate}
\end{corollary}
\begin{proof}
The finite entropies of $\{X_i\}_{i=1}^n$ assure that the
entropies involved in the sequences
\eqref{03.03.2021b3}--\eqref{eq2: 14.03.2021a} are finite.
Item~(a) follows from Theorem \ref{theorem: Methodology 1}{a}, where the submodular set functions
$f$ which correspond to \eqref{03.03.2021b3}--\eqref{eq1: 14.03.2021a}
are given in \eqref{entropic function}, \eqref{set function 3} and \eqref{set function 5},
respectively, and $g$ is the identity function on the real line.
The identity $k \binom{n}{k} = n \binom{n-1}{k-1}$ is used for \eqref{eq1: 14.03.2021a}.
Item~(b) follows from Theorem \ref{theorem: Methodology 1}{c}, where $f$ is the supermodular
function in \eqref{set function 2} and $g$ is the identity function on the real line.
We next prove Item~(c). The sequence \eqref{03.03.2021b4} is monotonically decreasing
by Theorem~\ref{theorem: Methodology 1}{a}, where $f$ is the submodular function
in \eqref{entropic function}, and $g \colon \Reals \to \Reals$ is the monotonically
increasing and convex function defined as $g(x) = \exp(2rx)$ for $x \in \Reals$ (with $r>0$).
The sequence \eqref{eq3: 15.03.2021a} is monotonically decreasing
by Theorem~\ref{theorem: Methodology 1}{d}, where $f$ is the supermodular
function in \eqref{set function 2}, and $g \colon \Reals \to \Reals$ is the monotonically
decreasing and convex function defined as $g(x) = \exp(-rx)$ for $x \in \Reals$.
The sequence \eqref{eq4: 15.03.2021a} is monotonically decreasing
by Theorem~\ref{theorem: Methodology 1}{a}, where $f$ is the submodular
function in \eqref{set function 3} and $g$ is the monotonically
increasing and convex function defined as $g(x) = \exp(rx)$ for $x \in \Reals$.
Finally, the sequence \eqref{eq2: 14.03.2021a} is monotonically decreasing
by Theorem~\ref{theorem: Methodology 1}{a}, where $f$ is the submodular
function in \eqref{set function 5} and $g$ is the monotonically
increasing and convex function defined as $g(x) = \exp(2rx)$ for $x \in \Reals$.
\end{proof}

\begin{remark}
From Proposition~\ref{proposition: information measures, continuous RVs}, since
the proof of Corollary~\ref{corollary: monotonicity - sums} only relies on the
sub/ supermodularity property of $f$, the random variables $\{X_i\}_{i=1}^n$
do not need to be discrete in Corollary~\ref{corollary: monotonicity - sums}.
In the reproduction of Han's inequality as an application of Corollary~\ref{corollary: 05.03.21b},
the random variables $\{X_i\}_{i=1}^n$ do not need to be discrete as well since
$f$ is not required to be nonnegative if $\alpha=1$ (only the submodularity
of $f$ in \eqref{entropic function} is required, which holds due to
Proposition~\ref{proposition: information measures, continuous RVs}).

In light of Proposition~\ref{proposition: information measures, continuous RVs}, since
the proof of Corollary~\ref{corollary: monotonicity - sums} only relies on the
submodularity/ supermodularity property of $f$, the random variables $\{X_i\}_{i=1}^n$
do not need to be discrete in Corollary~\ref{corollary: monotonicity - sums}.
In the reproduction of Han's inequality as an application of Corollary~\ref{corollary: 05.03.21b},
the random variables $\{X_i\}_{i=1}^n$ do not need to be discrete as well since the
requirement that $f$ be nonnegative is removed for $\alpha=1$ (only the submodularity
of $f$ in \eqref{entropic function} is required, which holds due to
Proposition~\ref{proposition: information measures, continuous RVs}).
\end{remark}

The following result exemplifies the utility of the monotonicity result of the sequence
\eqref{03.03.2021b3} in extremal combinatorics. It also generalizes the result in
Section~3.2 of \cite{Radhakrishnan01} for an achievable upper bound on the cardinality of a finite set in
the three-dimensional Euclidean space, expressed as a function of its number of projections on each of the
planes $XY, XZ$ and $YZ$. The next result provides an achievable upper bound on the cardinality of
a finite set of points in an $n$-dimensional Euclidean space, expressed as a function of its number
of projections on each of the $k$-dimensional Euclidean subspaces with an arbitrary $k < n$.

\begin{proposition}  \label{proposition: projections}
Let $\set{P} \subseteq \Reals^n$ be a finite set of points in the $n$-dimensional Euclidean space
with $\card{\set{P}} \eqdef M$. Let $k \in \OneTo{n-1}$, and $\ell \eqdef \binom{n}{k}$.
Let $\set{R}_1, \ldots, \set{R}_{\ell}$ be the projections of $\set{P}$ on each of the $k$-dimensional
subspaces of $\Reals^n$, and let $\card{\set{R}_j} = M_j \,$ for all $j \in \OneTo{\ell}$. Then,
\begin{eqnarray} \label{eq: 19.04.22b1}
\card{\set{P}} \leq \Biggl( \, \prod_{j=1}^{\binom{n}{k}} M_j \Biggr)^{\tfrac1{\tbinom{n-1}{k-1}}}.
\end{eqnarray}
Let $R \eqdef \frac{\log M}{n}$, and $R_j \eqdef \frac{\log M_j}{k}$ for all $j \in \OneTo{\ell}$.
An equivalent form of \eqref{eq: 19.04.22b1} is given by the inequality
\begin{eqnarray}  \label{eq: 19.04.22b2}
R \leq \frac1{\ell} \sum_{j=1}^\ell R_j.
\end{eqnarray}
Moreover, if $M_1 = \ldots = M_{\ell}$ and $\sqrt[k]{M_1} \in \naturals$, then
\eqref{eq: 19.04.22b1} and \eqref{eq: 19.04.22b2} are satisfied with equality
if $\set{P}$ is a grid of points in $\Reals^n$ with $\sqrt[k]{M_1}$ points on
each dimension (so, $M = M_1^{\frac{n}{k}}$).
\end{proposition}

\begin{proof}
Pick uniformly at random a point $X^n = (X_1, \ldots, X_n) \in \set{P}$. Then,
\begin{eqnarray} \label{eq: 19.04.22b3}
\Ent{X^n} = \log \card{\set{P}}.
\end{eqnarray}
The sequence in \eqref{03.03.2021b3} is monotonically decreasing, so
$h_k^{(n)} \geq h_n^{(n)}$, which is equivalent to
\begin{eqnarray} \label{eq: 19.04.22b4}
\binom{n-1}{k-1} \, \Ent{X^n} \leq \sum_{\set{T} \subseteq \OneTo{n}: \, \card{\set{T}}=k} \Ent{X_{\set{T}}}.
\end{eqnarray}
Let $\set{S}_1, \ldots, \set{S}_{\ell}$ be the $k$-subsets of the set $\OneTo{n}$, ordered in a way such that
$M_j$ is the cardinality of the projection of the set $\set{P}$ on the $k$-dimensional subspace whose coordinates
are the elements of the subset $\set{S}_j$. Then, \eqref{eq: 19.04.22b4} can be expressed in the form
\begin{eqnarray} \label{eq: 19.04.22b5}
\binom{n-1}{k-1} \, \Ent{X^n} \leq \sum_{j=1}^{\ell} \Ent{X_{\set{S}_j}},
\end{eqnarray}
and also
\begin{eqnarray} \label{eq: 19.04.22b6}
\Ent{X_{\set{S}_j}} \leq \log M_j, \quad j \in \OneTo{\ell},
\end{eqnarray}
since the entropy of a random variable is upper bounded by the logarithm of the number of its possible values.
Combining \eqref{eq: 19.04.22b3}, \eqref{eq: 19.04.22b5} and \eqref{eq: 19.04.22b6} gives
\begin{eqnarray} \label{eq: 19.04.22b7}
\binom{n-1}{k-1} \, \log \card{\set{P}} \leq \sum_{j=1}^{\ell} \log M_j.
\end{eqnarray}
Exponentiating both sides of \eqref{eq: 19.04.22b7} gives \eqref{eq: 19.04.22b1}. In addition, using
the identity $\binom{n}{k} = \frac{n}{k} \, \binom{n-1}{k-1}$ gives \eqref{eq: 19.04.22b2} from
\eqref{eq: 19.04.22b7}. Finally, the sufficiency condition for equalities in \eqref{eq: 19.04.22b1}
or \eqref{eq: 19.04.22b2} can be easily verified, which is obtained if $\set{P}$ is a grid of
points in $\Reals^n$ with the same finite number of projections on each dimension.
\end{proof}

\subsection{Connections to a Generalized Version of Shearer's Lemma and Other Results in the Literature}
\label{subsection: Connections to Literature}

The next proposition is a known generalized version of Shearer's Lemma.
\begin{proposition}
\label{proposition: Shearer - submodularity}
Let $\Omega$ be a finite set, let $\{\set{S}_j\}_{j=1}^M$ be a finite collection of
subsets of $\Omega$ (with $M \in \naturals$), and let $f \colon 2^{\Omega} \to \Reals$
be a set function.
\begin{enumerate}[a)]
\item If $f$ is non-negative and submodular, and every element in $\Omega$ is included
in at least $d \geq 1$ of the subsets $\{\set{S}_j\}_{j=1}^M$, then
\begin{eqnarray} \label{06.03.2021b1}
\sum_{j=1}^M f(\set{S}_j) \geq d \, f(\Omega).
\end{eqnarray}
\item If $f$ is a rank function, $\set{A} \subset \Omega$, and every element in
$\set{A}$ is included in at least $d \geq 1$ of the subsets $\{\set{S}_j\}_{j=1}^M$, then
\begin{eqnarray}  \label{06.03.2021b2}
\sum_{j=1}^M f(\set{S}_j) \geq d \, f(\set{A}).
\end{eqnarray}
\end{enumerate}
\end{proposition}
The first part of Proposition~\ref{proposition: Shearer - submodularity}
was pointed out in Section~1.5 of \cite{PolyanskiyW19}, and the second part
of Proposition~\ref{proposition: Shearer - submodularity} is a
generalization of Remark~1 and inequality~(47) in \cite{Sason_Entropy21}.
Appendix~\ref{Appendix: Proof - Shearer - submodularity} provides
a (somewhat different) proof of Proposition~\ref{proposition: Shearer - submodularity}{a},
as well as a self-contained proof of Proposition~\ref{proposition: Shearer - submodularity}{b}.

\bigskip
Let $\{X_i\}_{i=1}^n$ be discrete random variables, and
consider the set function $f \colon 2^{\OneTo{n}} \to \Reals_{+}$ which is defined as
$f(\set{A}) = \Ent{X_{\set{A}}}$ for all $\set{A} \subseteq \OneTo{n}$. Since $f$
is a rank function \cite{Fujishige78}, Proposition~\ref{proposition: Shearer - submodularity}
then specializes to Shearer's Lemma \cite{ChungGFS86} and a modified version of this lemma
Remark~1 of \cite{Sason_Entropy21}.

In light of Item~e) in Proposition~\ref{proposition: information measures, polymatroids}
and Item~b) of Proposition~\ref{proposition: Shearer - submodularity},
Corollaries~\ref{corollary: Shearer - sum entropy} and~\ref{corollary: entropies of k-element sums}
are obtained as follows.

\begin{corollary}  \label{corollary: Shearer - sum entropy}
Let $\{X_i\}_{i=1}^n$ be independent discrete random variables, $\{\set{S}_j\}_{j=1}^M$
be subsets of $\OneTo{n}$, and $\set{A} \subseteq \OneTo{n}$.
If each element in $\set{A}$ belongs to at least $d \geq 1$ of the sets
$\{\set{S}_j\}_{j=1}^M$,~then
\begin{eqnarray}  \label{eq1: Shearer - sum entropy}
d \, \BiggEnt{ \, \sum_{i \in \set{A}} X_i} \leq
\sum_{j=1}^M \BiggEnt{ \, \sum_{i \in \set{S}_j} X_i}.
\end{eqnarray}
In particular, if every $i \in \OneTo{n}$ is included in at least $d \geq 1$ of the subsets
$\{\set{S}_j\}_{j=1}^M$, then
\begin{eqnarray}  \label{eq2: Shearer - sum entropy}
d \, \BiggEnt{ \, \sum_{i=1}^n  X_i} \leq
\sum_{j=1}^M \BiggEnt{ \, \sum_{i \in \set{S}_j} X_i}.
\end{eqnarray}
\end{corollary}

\begin{remark}
Inequality \eqref{eq2: Shearer - sum entropy} is also a special case of \cite[Theorem~2]{Madiman_ITW08},
and they coincide if every element $i \in \OneTo{n}$ is included in a fixed number ($d$) of the subsets
$\{\set{S}_j\}_{j=1}^M$.
\end{remark}

A specialization of Corollary~\ref{corollary: Shearer - sum entropy} gives the next result.
\begin{corollary}  \label{corollary: entropies of k-element sums}
Let $\{X_i\}_{i=1}^n$ be independent and discrete random variables with finite variances. Then,
the following holds:
\begin{enumerate}[a)]
\item For every $k \in \OneTo{n-1}$,
\begin{eqnarray}  \label{eq3: 14.03.2021a}
\BiggEnt{\, \sum_{i=1}^n X_i} \leq \frac1{\binom{n-1}{k-1}}
\sum_{\set{T} \subseteq \OneTo{n}: \, \card{\set{T}} = k}
\BiggEnt{\, \sum_{\omega \in \set{T}} X_\omega},
\end{eqnarray}
and equivalently,
\begin{eqnarray}  \label{eq4: 14.03.2021a}
\BiggEntPow{\sum_{i=1}^n X_i} \leq
\Biggl\{\prod_{\set{T} \subseteq \OneTo{n}: \, \card{\set{T}} = k}
\BiggEntPow{\, \sum_{\omega \in \set{T}} X_\omega}\Biggr\}^{\frac1{\binom{n-1}{k-1}}}.
\end{eqnarray}
\item For every $k \in \OneTo{n-1}$,
\begin{eqnarray}  \label{eq5: 14.03.2021a}
\BiggEntPow{\sum_{i=1}^n X_i} \leq
\frac1{\binom{n}{k}} \sum_{\set{T} \subseteq \OneTo{n}: \, \card{\set{T}} = k}
\BiggEntPowToPow{\frac{n}{k}}{\, \sum_{\omega \in \set{T}} X_\omega},
\end{eqnarray}
where \eqref{eq5: 14.03.2021a} is in general looser than \eqref{eq4: 14.03.2021a}, with
equivalence if $\{X_i\}_{i=1}^n$ are i.i.d.; in particular,

\vspace*{-0.3cm}
\begin{subequations} \label{eq6: 14.03.2021a}
\begin{align}
\BiggEntPow{\sum_{i=1}^n X_i} &\leq
\Biggl\{\prod_{j=1}^n \BiggEntPow{\, \sum_{i \neq j} X_i} \Biggr\}^{\frac1{n-1}} \\
&\leq \frac1n \sum_{j=1}^n \Biggl\{\BiggEntPow{\sum_{i \neq j} X_i}\Biggr\}^{\frac{n}{n-1}}.
\end{align}
\end{subequations}
\end{enumerate}
\end{corollary}
\begin{proof}
Let $\{\set{S}_j\}_{j=1}^M$ be all the $k$-element subsets of $\Omega = \OneTo{n}$ (with $M = \binom{n}{k}$).
Then, every element $i \in \OneTo{n}$ belongs to $d = \frac{kM}{n} = \binom{n-1}{k-1}$ such subsets, which
then gives \eqref{eq3: 14.03.2021a} as a special case of \eqref{eq2: Shearer - sum entropy}. Alternatively,
\eqref{eq3: 14.03.2021a} follows from Corollary~\ref{corollary: monotonicity - sums}b, which yields
$m_k^{(n)} \geq m_n^{(n)}$ for all $k \in \OneTo{n-1}$.
Exponentiating both sides of \eqref{eq3: 14.03.2021a} gives \eqref{eq4: 14.03.2021a}.
Inequality~\eqref{eq5: 14.03.2021a} is a loosened version of \eqref{eq4: 14.03.2021a},
which follows by invoking the AM-GM inequality (i.e., the geometric mean of
nonnegative real numbers is less than or equal to their arithmetic mean,
with equality between these two means if and only if these numbers are all equal),
in conjunction with the identity $\frac{k}{n} \binom{n}{k} = \binom{n-1}{k-1}$.
Inequalities \eqref{eq4: 14.03.2021a} and \eqref{eq5: 14.03.2021a} are consequently equivalent
if $\{X_i\}_{i=1}^n$ are i.i.d. random variables, and \eqref{eq6: 14.03.2021a} is a specialized
version of \eqref{eq4: 14.03.2021a} and the loosened inequality \eqref{eq5: 14.03.2021a} by setting $k=n-1$.
\end{proof}

The next remarks consider information inequalities in
Corollaries~\ref{corollary: monotonicity - sums}--\ref{corollary: entropies of k-element sums},
in light of Theorem~\ref{theorem: Methodology 1} here, and some known results in the literature.
\begin{remark}
Inequality \eqref{eq3: 14.03.2021a} was derived by Madiman as a special case
of Theorem~2 in \cite{Madiman_ITW08}.
The proof of Corollary~\ref{corollary: entropies of k-element sums}{a} shows that
\eqref{eq3: 14.03.2021a} can be also derived in two different ways as
special cases of both Theorem~\ref{theorem: Methodology 1}{a} and
Proposition~\ref{proposition: Shearer - submodularity}{a}.
\end{remark}

\begin{remark}
Inequality \eqref{eq5: 14.03.2021a} can be also derived as a special case of
Theorem~\ref{theorem: Methodology 1}{a}, where $f$ is the rank function in \eqref{set function 5},
and $g \colon \Reals \to \Reals$ is given by $g(x) \eqdef \exp(2nx)$ for all $x \in \Reals$.
It also follows from the monotonicity property in Corollary~\ref{corollary: monotonicity - sums}{c},
which yields $w_k^{(n)}(n) \geq w_n^{(n)}(n)$ for all $k \in \OneTo{n-1}$.
\end{remark}

\begin{remark}
The result in Theorem~8 of \cite{KIshiOY_ISIT2014} is a special case of Theorem \ref{theorem: Methodology 1}{a}
here, which follows by taking the function $g$ in Theorem \ref{theorem: Methodology 1}{a} to be the identity function.
The flexibility in selecting the function $g$ in Theorem \ref{theorem: Methodology 1} enables to obtain a larger
collection of information inequalities. This is in part reflected from a comparison of
Corollary~\ref{corollary: monotonicity - sums} here with Corollary~9 of \cite{KIshiOY_ISIT2014}.
More specifically, the findings about the monotonicity properties in \eqref{03.03.2021b3}, \eqref{eq2: 15.03.2021a}
and \eqref{eq1: 15.03.2021a} were obtained in Corollary~9 of \cite{KIshiOY_ISIT2014}, while relying on Theorem 8 of
\cite{KIshiOY_ISIT2014} and the sub/supermodularity properties of the considered Shannon information measures. It
is noted, however, that the monotonicity results of the sequences~\eqref{03.03.2021b4}--\eqref{eq2: 14.03.2021a}
(Corollary \ref{corollary: monotonicity - sums}{c}) are not implied by Theorem~8 of \cite{KIshiOY_ISIT2014}.
\end{remark}

\begin{remark}
Inequality~\eqref{eq6: 14.03.2021a} forms a counterpart of an entropy power inequality by
Artstein {\em et al.} (Theorem~3 of \cite{ArtsteinBBN04}), where for independent
random variables $\{X_i\}_{i=1}^n$ with finite variances:

\vspace*{-0.5cm}
\begin{eqnarray}  \label{ABBN inequality}
\BiggEntPow{\sum_{i=1}^n X_i} \geq
\frac1{n-1} \sum_{j=1}^n \, \BiggEntPow{\sum_{i \neq j} X_i}.
\end{eqnarray}
Inequality \eqref{eq4: 14.03.2021a}, and also its looser version in \eqref{eq5: 14.03.2021a},
form counterparts of the generalized inequality by Madiman and Barron, which reads (see
inequality~(4) in \cite{MadimanB07}):

\vspace*{-0.5cm}
\begin{eqnarray} \label{Madiman-Barron inequality}
\BiggEntPow{\sum_{i=1}^n X_i} \geq
\frac1{\binom{n-1}{k-1}} \, \sum_{\set{T} \subseteq \OneTo{n}: \, \card{\set{T}} = k}
\BiggEntPow{\sum_{\omega \in \set{T}} X_\omega},  \quad k \in \OneTo{n-1}.
\end{eqnarray}
\end{remark}

\section{Proofs}
\label{section: proofs}
The present section provides proofs of (most of the) results in Section~\ref{section: methodologies}.

\subsection{Proof of Theorem~\ref{theorem: Methodology 1}}
\label{subsection: Proof of Methodology 1}

We prove Item~a, and then readily prove Items~b--d.
Define the auxiliary sequence
\begin{eqnarray}
\label{01.03.2021a1}
f_k^{(n)} \eqdef \frac1{\binom{n}{k}} \sum_{\set{T} \subseteq \Omega: \, \card{\set{T}}=k}
f(\set{T}), \qquad k \in \FromTo{0}{n},
\end{eqnarray}
averaging $f$ over all $k$-element
subsets of the $n$-element set $\Omega \eqdef \{\omega_1, \ldots, \omega_n\}$.
Let the permutation $\pi \colon \OneTo{n} \to \OneTo{n}$ be arbitrary.
For $k \in \OneTo{n-1}$, let
\begin{subequations}
\begin{eqnarray}
&& \set{S}_1 \eqdef \{\omega_{\pi(1)}, \ldots, \omega_{\pi(k-1)}, \omega_{\pi(k)} \},   \label{eq1: 01.03.2021c1} \\[0.1cm]
&& \set{S}_2 \eqdef \{\omega_{\pi(1)}, \ldots, \omega_{\pi(k-1)}, \omega_{\pi(k+1)} \}, \label{eq2: 01.03.2021c1}
\end{eqnarray}
\end{subequations}
which are $k$-element subsets of $\Omega$ with $k-1$ elements in common. Then,

\vspace*{-0.3cm}
\begin{align}
\label{01.03.2021c2}
f(\set{S}_1) + f(\set{S}_2) \geq f(\set{S}_1 \cup \set{S}_2) + f(\set{S}_1 \cap \set{S}_2),
\end{align}
which holds by the submodularity of~$f$ (by assumption), i.e.,

\vspace*{-0.3cm}
\begin{align}
\label{01.03.2021c3}
& f\bigl( \{\omega_{\pi(1)}, \ldots, \omega_{\pi(k)} \} \bigr)
+ f\bigl( \{\omega_{\pi(1)}, \ldots, \omega_{\pi(k-1)}, \omega_{\pi(k+1)} \} \bigr) \nonumber \\[0.1cm]
& \quad \geq f\bigl( \{\omega_{\pi(1)}, \ldots, \omega_{\pi(k+1)} \}\bigr)
+ f\bigl( \{\omega_{\pi(1)}, \ldots, \omega_{\pi(k-1)} \} \bigr).
\end{align}
Averaging the terms on both sides of \eqref{01.03.2021c3} over all the $n!$ permutations $\pi$
of $\OneTo{n}$ gives

\vspace*{-0.3cm}
\begin{subequations} \label{01.03.2021c4}
\begin{align}
\frac1{n!} \sum_{\pi} f\bigl( \{\omega_{\pi(1)}, \ldots, \omega_{\pi(k)} \} \bigr)
&= \frac{k! \, (n-k)!}{n!} \sum_{\set{T} \subseteq \Omega: \, \card{\set{T}} = k}  f(\set{T}) \\
&= \frac{1}{\binom{n}{k}} \sum_{\set{T} \subseteq \Omega: \, \card{\set{T}} = k}  f(\set{T}) \\
&= f_k^{(n)},
\end{align}
\end{subequations}
and similarly
\begin{subequations}  \label{01.03.2021c5}
\begin{eqnarray}
&& \frac1{n!} \sum_{\pi} f\bigl( \{\omega_{\pi(1)}, \ldots, \omega_{\pi(k-1)}, \omega_{\pi(k+1)} \} \bigr) = f_k^{(n)}, \\[0.1cm]
&& \frac1{n!} \sum_{\pi} f\bigl( \{\omega_{\pi(1)}, \ldots, \omega_{\pi(k+1)} \}\bigr) = f_{k+1}^{(n)}, \\[0.1cm]
&& \frac1{n!} \sum_{\pi} f\bigl( \{\omega_{\pi(1)}, \ldots, \omega_{\pi(k-1)} \}\bigr) = f_{k-1}^{(n)},
\end{eqnarray}
\end{subequations}
with $f_0^{(n)} = 0$ since by assumption $f(\es) = 0$.
Combining \eqref{01.03.2021c3}--\eqref{01.03.2021c5} gives
\begin{eqnarray} \label{01.03.2021a3}
2 f_k^{(n)} \geq f_{k+1}^{(n)} + f_{k-1}^{(n)}, \quad k \in \OneTo{n-1},
\end{eqnarray}
which is rewritten as
\begin{eqnarray}
\label{01.03.2021c8}
f_k^{(n)} - f_{k-1}^{(n)} \geq f_{k+1}^{(n)} - f_k^{(n)}, \quad k \in \OneTo{n-1}.
\end{eqnarray}
Consequently, it follows that

\vspace*{-0.3cm}
\begin{subequations} \label{01.03.2021c9}
\begin{align}
\label{eq1: 01.03.2021c9}
\frac{f_k^{(n)}}{k} - \frac{f_{k+1}^{(n)}}{k+1} &=
\frac1k \sum_{j=1}^k \Bigl( f_j^{(n)} - f_{j-1}^{(n)} \Bigr)
- \frac1{k+1} \sum_{j=1}^{k+1} \Bigl( f_j^{(n)} - f_{j-1}^{(n)} \Bigr) \\
\label{eq2: 01.03.2021c9}
&= \biggl( \frac1k - \frac1{k+1} \biggr) \sum_{j=1}^k \Bigl( f_j^{(n)}
- f_{j-1}^{(n)} \Bigr) \, - \frac1{k+1} \, \Bigl( f_{k+1}^{(n)} - f_k^{(n)} \Bigr) \\
\label{eq3: 01.03.2021c9}
&= \frac1{k(k+1)} \sum_{j=1}^k \Bigl\{ \Bigl( f_j^{(n)} - f_{j-1}^{(n)} \Bigr)
- \Bigl( f_{k+1}^{(n)} - f_k^{(n)} \Bigr) \Bigr\} \\
\label{eq4: 01.03.2021c9}
&\geq 0,
\end{align}
\end{subequations}
where equality \eqref{eq1: 01.03.2021c9} holds since $f_0^{(n)}=0$, and
inequality \eqref{eq4: 01.03.2021c9} holds by \eqref{01.03.2021c8}.
The sequence $\Bigl\{ \frac{f_k^{(n)}}{k} \Bigr\}_{k=1}^n$ is therefore monotonically
decreasing, and in particular
\begin{eqnarray}  \label{01.03.2021c10}
f_k^{(n)} \geq \frac{k \, f_n^{(n)}}{n} = \frac{k}{n}.
\end{eqnarray}

We next prove \eqref{01.03.2021b3} when $\alpha=1$, and
then proceed to prove Theorem~\ref{theorem: Methodology 1}. By \eqref{01.03.2021c10}
\begin{eqnarray}  \label{01.03.2021c11}
\frac{f_n^{(n)}}{n} \leq \frac{f_{n-1}^{(n)}}{n-1},
\end{eqnarray}
where, by \eqref{01.03.2021a1},
\begin{eqnarray}  \label{01.03.2021c12}
f_n^{(n)} = f(\Omega), \qquad f_{n-1}^{(n)}
= \frac1n \sum_{\set{T} \subseteq \Omega: \, \card{\set{T}} = n-1}  f(\set{T}).
\end{eqnarray}
Combining \eqref{01.03.2021c11} and \eqref{01.03.2021c12} gives
\begin{eqnarray}  \label{01.03.2021c13}
(n-1) \, f(\Omega) \leq\sum_{\set{T} \subseteq \Omega: \, \card{\set{T}} = n-1}  f(\set{T}).
\end{eqnarray}
Since there are $n$ subsets $\set{T} \subseteq \Omega$ with $\card{\set{T}} = n-1$,
rearranging terms in \eqref{01.03.2021c13} gives \eqref{01.03.2021b3} for
$\alpha=1$; it is should be noted that, for $\alpha=1$, the set function
$f$ does not need to be nonnegative for the satisfiability of \eqref{01.03.2021b3}
(however, this will be required for $\alpha>1$).

We next prove Item~a). By \eqref{01.03.2021b1}, for $k \in \OneTo{n}$,

\vspace*{-0.3cm}
\begin{subequations}  \label{01.03.2021d1}
\begin{align}
\label{eq1: 01.03.2021d1}
t_k^{(n)} &= \frac1{\binom{n}{k}} \sum_{\set{T} \subseteq \Omega: \, \card{\set{T}}=k}
g\biggl(\frac{f(\set{T})}{k} \biggr) \\
\label{eq2: 01.03.2021d1}
&= \frac1{\binom{n}{k}} \sum_{\set{T} = \{t_1, \ldots, t_k\} \subseteq \Omega}
g\Biggl(\frac{f\bigl(\{t_1, \ldots, t_k\}\bigr)}{k} \Biggr).
\end{align}
\end{subequations}
Fix $\Omega_k \eqdef \{t_1, \ldots, t_k\} \subseteq \Omega$, and let $\tilde{f} \colon 2^{\Omega_k} \to \Reals$
be the restriction of the function $f$ to the subsets of $\Omega_k$.
Then, $\tilde{f}$ is a submodular set function with $\tilde{f}(\es) = 0$; similarly to \eqref{01.03.2021a1},
\eqref{01.03.2021c11} and \eqref{01.03.2021c12} with $f$ replaced by $\tilde{f}$, and $n$ replaced by $k$,
the sequence $\Bigl\{ \frac{\tilde{f}_j^{(k)}}{j} \Bigr\}_{j=1}^k$ is monotonically decreasing. Hence,
for $k \in \FromTo{2}{n}$,
\begin{eqnarray}  \label{01.03.2021d2}
\frac{\tilde{f}_k^{(k)}}{k} \leq \frac{\tilde{f}_{k-1}^{(k)}}{k-1},
\end{eqnarray}
where

\vspace*{-0.3cm}
\begin{subequations}  \label{01.03.2021d3}
\begin{align}
\label{eq1: 01.03.2021d3}
\tilde{f}_k^{(k)} &= \tilde{f}(\Omega_k) = f\bigl(\{t_1, \ldots, t_k\}\bigr), \\
\label{eq2: 01.03.2021d3}
\tilde{f}_{k-1}^{(k)} &= \frac1k \sum_{\set{T} \subseteq \Omega_k: \, \card{\set{T}} = k-1}
\tilde{f}(\set{T}) \\
\label{eq3: 01.03.2021d3}
&= \frac1k \sum_{\set{T} \subseteq \Omega_k: \, \card{\set{T}} = k-1}  f(\set{T}) \\
\label{eq4: 01.03.2021d3}
&= \frac1k \sum_{i=1}^k f\bigl(\{t_1, \ldots, t_{i-1}, t_{i+1}, \ldots, t_k\} \bigr).
\end{align}
\end{subequations}
Combining \eqref{01.03.2021d2} and \eqref{01.03.2021d3} gives
\begin{eqnarray}  \label{01.03.2021d4}
f\bigl(\{t_1, \ldots, t_k\}\bigr) \leq \frac1{k-1} \sum_{i=1}^k
f\bigl(\{t_1, \ldots, t_{i-1}, t_{i+1}, \ldots, t_k\} \bigr),
\end{eqnarray}
and, since by assumption $g$ is monotonically increasing,
\begin{eqnarray}  \label{01.03.2021d5}
g\Biggl(\frac{f\bigl(\{t_1, \ldots, t_k\}\bigr)}{k} \Biggr)
\leq g\Biggl( \frac1{k} \sum_{i=1}^k \frac{f\bigl(\{t_1, \ldots,
t_{i-1}, t_{i+1}, \ldots, t_k\} \bigr)}{k-1} \Biggr).
\end{eqnarray}
From \eqref{01.03.2021d1} and \eqref{01.03.2021d5}, for all $k \in \FromTo{2}{n}$,
\begin{eqnarray} \label{01.03.2021d6}
t_k^{(n)} \leq \frac1{\binom{n}{k}} \sum_{\set{T} = \{t_1, \ldots, t_k\} \subseteq \Omega}
g\Biggl( \frac1{k} \sum_{i=1}^k \frac{f\bigl(\{t_1, \ldots, t_{i-1}, t_{i+1},
\ldots, t_k\} \bigr)}{k-1} \Biggr),
\end{eqnarray}
and

\vspace*{-0.3cm}
\begin{subequations} \label{01.03.2021d7}
\begin{align}
\label{eq1: 01.03.2021d7}
t_k^{(n)} & \leq \frac1{k \binom{n}{k}}  \sum_{i=1}^k \sum_{\set{T}
= \{t_1, \ldots, t_k\} \subseteq \Omega}
g\Biggl( \frac{f\bigl(\{t_1, \ldots, t_{i-1}, t_{i+1}, \ldots, t_k\} \bigr)}{k-1} \Biggr) \\[0.1cm]
\label{eq2: 01.03.2021d7}
&= \frac{n-k+1}{k \binom{n}{k}}  \sum_{i=1}^k \Biggl\{
\sum_{\{t_1, \ldots, t_{i-1}, t_{i+1}, \ldots, t_k\} \subseteq \Omega}
g\Biggl( \frac{f\bigl(\{t_1, \ldots, t_{i-1}, t_{i+1}, \ldots, t_k\} \bigr)}{k-1}
\Biggr) \Biggr\} \\[0.1cm]
\label{eq3: 01.03.2021d7}
&= \frac{k! \, (n-k)! \, (n-k+1)}{n! \; k} \sum_{\set{S} \subseteq \Omega: \, \card{\set{S}} = k-1}
g\biggl(\frac{f(\set{S})}{k-1}\biggr) \\[0.1cm]
\label{eq4: 01.03.2021d7}
&= \frac{(k-1)! \, (n-k+1)!}{n!} \sum_{\set{S} \subseteq \Omega: \, \card{\set{S}} = k-1}
g\biggl(\frac{f(\set{S})}{k-1}\biggr) \\[0.1cm]
\label{eq5: 01.03.2021d7}
&= \frac1{\binom{n}{k-1}} \, \sum_{\set{S} \subseteq \Omega: \, \card{\set{S}} = k-1}
g\biggl(\frac{f(\set{S})}{k-1}\biggr) \\[0.1cm]
\label{eq6: 01.03.2021d7}
&= t_{k-1}^{(n)},
\end{align}
\end{subequations}
where \eqref{eq1: 01.03.2021d7} holds by invoking Jensen's inequality to the convex function $g$;
\eqref{eq2: 01.03.2021d7} holds since the term of the inner summation in the right-hand side of
\eqref{eq1: 01.03.2021d7} does not depend on $t_i$, so for every $(k-1)$-element subset
$\set{S} = \{t_1, \ldots, t_{i-1}, t_{i+1}, \ldots, t_k\} \subseteq \Omega$, there are $n-k+1$
possibilities to extend it by a single element $(t_i)$ into a $k$-element subset
$\set{T} = \{t_1, \ldots, t_k\} \subseteq \Omega$;
\eqref{eq5: 01.03.2021d7} is straightforward, and
\eqref{eq6: 01.03.2021d7} holds by the definition in \eqref{01.03.2021b1}.
This proves Item~a).

Item~b) follows from Item~a), and similarly Item~d) follows from Item~c), by replacing $g$
with $-g$. Item~c) is next verified.
If $f$ is a supermodular set function with $f(\es)~=~0$, then \eqref{01.03.2021c2},
\eqref{01.03.2021c3}, and \eqref{01.03.2021a3}--\eqref{01.03.2021c9} hold with flipped
inequality signs.
Hence, if $g$ is monotonically decreasing, then inequalities \eqref{01.03.2021d5}
and \eqref{01.03.2021d6} are reversed; finally, if $g$ is also concave, then (by Jensen's inequality)
\eqref{01.03.2021d7} holds with a flipped inequality sign, which proves Item~c).

\subsection{Proof of Corollary~\ref{corollary: 05.03.21a}}
\label{subsection: proof of corollary 1}

By assumption $f \colon 2^{\Omega} \to \Reals$ is a rank function, which implies that
$0 \leq f(\set{T}) \leq f(\Omega)$ for every $\set{T} \subseteq \Omega$.
Since (by definition) $f$ is submodular with $f(\es) = 0$, and (by assumption)
the function $g$ is convex and monotonically increasing, then (from
\eqref{04.03.2021a1}, while replacing $k$ with $k_n$)
\begin{eqnarray}
\label{04.03.2021a3}
\binom{n}{k_n} \; g\biggl(\frac{f(\Omega)}{n} \biggr) \leq
\sum_{\set{T} \subseteq \Omega: \, \card{\set{T}}=k_n} g\biggl(\frac{f(\set{T})}{k_n} \biggr)
\leq \binom{n}{k_n} \; g\biggl(\frac{f(\Omega)}{k_n} \biggr), \qquad n \in \naturals.
\end{eqnarray}
By the second assumption in Corollary~\ref{corollary: 05.03.21a}, for positive
values of $x$ that are sufficiently close to zero, we have
\begin{itemize}
\item $g(x) \approx g(0) > 0$ if $g(0)>0$;
\item
$g(x)$ scales like $\frac1{\ell !} \, g^{(\ell)}(0) \, x^\ell$ if
$g(0) = \ldots = g^{(\ell-1)}(0) = 0$ with $g^{(\ell)}(0) > 0$ for some $\ell \in \naturals$.
\end{itemize}
In both cases, it follows that
\begin{eqnarray}
\label{20.03.2021a1}
\lim_{x \to 0^+} x \log g(x) = 0.
\end{eqnarray}
In light of \eqref{04.03.2021a3} and \eqref{20.03.2021a1}, and since (by assumption) $k_n \underset{n \to \infty}{\longrightarrow} \infty$, it follows that
\begin{eqnarray} \label{19.04.22a1}
\lim_{n \to \infty} \frac1n \, \Biggl[ \log \Biggl( \sum_{\set{T} \subseteq \Omega:
\, \card{\set{T}}=k_n} g\biggl(\frac{f(\set{T})}{k_n} \biggr) \Biggr) - \log \binom{n}{k_n} \Biggr] = 0.
\end{eqnarray}
By the following upper and lower bounds on the binomial coefficient:
\begin{eqnarray}
\label{01.03.2021d15}
\frac1{n+1} \, \exp\biggl(n \, \biggEntBin{\frac{k_n}{n}} \biggr) \leq \binom{n}{k_n}
\leq \exp\biggl(n \, \biggEntBin{\frac{k_n}{n}} \biggr),
\end{eqnarray}
the combination of equalities \eqref{19.04.22a1} and \eqref{01.03.2021d15} gives equality \eqref{19.04.22a2}.
Equality \eqref{04.03.2021a2} holds as a special case of \eqref{19.04.22a2},
under the assumption that $\underset{n \to \infty}{\lim} \frac{k_n}{n} = \beta \in [0,1]$.

\subsection{Proof of Corollary~\ref{corollary: 05.03.21b}}
\label{subsection: proof of corollary 2}

For $\alpha=1$, Corollary~\ref{corollary: 05.03.21b} is proved in \eqref{01.03.2021c13}.
Fix $\alpha>1$, and let $g \colon \Reals \to \Reals$ be
\begin{eqnarray}  \label{01.03.2021d8}
g(x) \eqdef
\begin{dcases}
x^\alpha, & x \geq 0, \\
0, & x<0,
\end{dcases}
\end{eqnarray}
which is monotonically increasing and convex on the real line.
By Theorem~\ref{theorem: Methodology 1}{a},
\begin{eqnarray}  \label{01.03.2021d9}
t_k^{(n)} \geq t_n^{(n)}, \quad k \in \OneTo{n}.
\end{eqnarray}
Since by assumption $f$ is nonnegative, it follows from \eqref{01.03.2021b1} and \eqref{01.03.2021d8} that

\vspace*{-0.3cm}
\begin{subequations}  \label{01.03.2021d10}
\begin{align}
\label{eq1: 01.03.2021d10}
t_k^{(n)} &= \frac1{\binom{n}{k}} \sum_{\set{T} \subseteq \Omega: \, \card{\set{T}}=k}
g\biggl(\frac{f(\set{T})}{k} \biggr) \\[0.1cm]
\label{eq2: 01.03.2021d10}
&= \frac1{k^\alpha \, \binom{n}{k}} \sum_{\set{T} \subseteq \Omega: \, \card{\set{T}}=k}  f^\alpha(\set{T}).
\end{align}
\end{subequations}
Combining \eqref{01.03.2021d9}--\eqref{01.03.2021d10} and rearranging terms gives, for all $\alpha > 1$,

\vspace*{-0.3cm}
\begin{subequations}  \label{01.03.2021d11}
\begin{align}
\label{eq1: 01.03.2021d11}
\sum_{\set{T} \subseteq \Omega: \, \card{\set{T}}=k}  f^\alpha(\set{T}) &\geq \Bigl(\frac{k}{n}\Bigr)^\alpha
\, \binom{n}{k} \, f^\alpha(\Omega) \\[0.1cm]
\label{eq2: 01.03.2021d11}
&= \Bigl(\frac{k}{n}\Bigr)^{\alpha-1} \, \binom{n-1}{k-1} \, f^\alpha(\Omega),
\end{align}
\end{subequations}
where equality \eqref{eq2: 01.03.2021d11} holds by the identity
$\frac{k}{n} \, \binom{n}{k} = \binom{n-1}{k-1}$. This further gives

\vspace*{-0.3cm}
\begin{subequations}
\begin{align}
\label{eq1: 01.03.2021d13}
\sum_{\set{T} \subseteq \Omega: \, \card{\set{T}}=k}  \bigl(f^\alpha(\Omega) - f^\alpha(\set{T}) \bigr)
&= \binom{n}{k} \, f^\alpha(\Omega) - \sum_{\set{T} \subseteq \Omega: \, \card{\set{T}}=k} f^\alpha(\set{T})  \\
\label{eq12: 01.03.2021d13}
&\leq \biggl(1 - \frac{k^\alpha}{n^\alpha}\biggr) \, \binom{n}{k} \, f^\alpha(\Omega) \\[0.1cm]
\label{eq3: 01.03.2021d13}
&= c_{\alpha}(n,k) \, f^\alpha(\Omega),
\end{align}
\end{subequations}
where equality \eqref{eq3: 01.03.2021d13} holds by the definition in \eqref{eq: c}.
This proves \eqref{01.03.2021b3} for $\alpha>1$.

We next prove Item~b).
The function $f$ is (by assumption) a rank function, which yields its nonnegativity. Hence,
the leftmost inequality in \eqref{01.03.2021b4} holds by \eqref{01.03.2021d11}. The rightmost
inequality in \eqref{01.03.2021b4} also holds since $f \colon 2^\Omega \to \Reals$ is
monotonically increasing, which yields $f(\set{T}) \leq f(\Omega)$ for all
$\set{T} \subseteq \Omega$. For $k \in \OneTo{n}$ and $\alpha \geq 0$ (in particular,
for $\alpha \geq 1$),
\begin{eqnarray}  \label{01.03.2021d14}
\sum_{\set{T} \subseteq \Omega: \, \card{\set{T}}=k}  f^\alpha(\set{T})
\leq \binom{n}{k} \, f^\alpha(\Omega),
\end{eqnarray}
where \eqref{01.03.2021d14} holds since there are $\binom{n}{k}$ $k$-element subsets
$\set{T}$ of the $n$-element set $\Omega$, and every summand $f^\alpha(\set{T})$ (with
$\set{T} \subseteq \Omega$) is upper bounded by $f^\alpha(\Omega)$.

\bigskip
\section{A Problem in Extremal Graph Theory}
\label{section: problem}

This section applies the generalization of Han's inequality in \eqref{eq1: Han's inequality}
to the following problem.

\subsection{Problem Formulation}
\label{subsection: problem formulation}

Let $\set{A} \subseteq \{-1, 1\}^n$, with $n \in \naturals$, and let $\tau \in \OneTo{n}$.
Let $G = G_{\set{A}, \tau}$ be an un-directed simple graph with vertex set $\V{G} = \set{A}$, and
pairs of vertices in $G$ are adjacent (i.e., connected by an edge) if and only if they are represented
by vectors in $\set{A}$ whose Hamming distance is less than or equal to $\tau$:
\begin{eqnarray}  \label{eq: 22032022a0}
\{x^n, y^n \} \in \E{G} \, \Leftrightarrow \, \bigl(x^n, y^n \in \set{A}, \; \; x^n \neq y^n, \; \; \dH{x^n}{y^n} \leq \tau \bigr).
\end{eqnarray}
The question is how large can the size of $G$ be (i.e., how many edges it may have) as a function of the cardinality
of the set $\set{A}$, and possibly based also on some basic properties of the set $\set{A}$ ?

This problem and its related analysis generalize and refine, in a nontrivial way, the bound in
Theorem~4.2 of \cite{Boucheron_Lugosi_Massart_book} which applies to the special case where $\tau=1$.
The motivation for this extension is next considered.

\subsection{Problem Motivation}
\label{subsection: problem motivation}

Constraint coding is common in many data recording systems and data communication systems, where
some sequences are more prone to error than others, and a constraint on the sequences that are allowed to be
recorded or transmitted is imposed in order to reduce the likelihood of error. Given such a constraint, it is
then necessary to encode arbitrary user sequences into sequences that obey the constraint.

From an information-theoretic perspective, this problem can be interpreted as follows.
Consider a communication channel $\chlong{W}{\set{X}}{\set{Y}}$ with input alphabet $\set{X}$ and output alphabet
$\set{Y}$, and suppose that a constraint is imposed on the sequences that are allowed to be transmitted over the
channel.  As a result of such a constraint, the information sequences are first encoded into codewords by an
error-correction encoder, followed by a constrained encoder that maps these codewords into constrained sequences.
Let them be binary $n$-length sequences from the set $\set{A} \subseteq \{-1, 1\}^n$. A channel modulator then
modulates these sequences into symbols from $\set{X}$, and the received sequences at the channel output, with
alphabet $\set{Y}$, are first demodulated, and then decoded (in a reverse order of the encoding process)
by the constrained decoder and error-correction decoder.

Consider a channel model where pairs of binary $n$-length sequences
from the set $\set{A}$ whose Hamming distance is less than or equal to a fixed number $\tau$ share a common output
sequence with positive probability, whereas this halts to be the case if the Hamming distance is larger than $\tau$.
In other words, we assume that by design, pairs of sequences in $\set{A}$ whose Hamming distance is larger than $\tau$
cannot be confused in the sense that there does not exist a common output sequence which may be possibly received
(with positive probability) at the channel output.

The confusion graph $G$ that is associated with this setup is an undirected simple graph whose vertices
represent the $n$-length binary sequences in $\set{A}$, and pairs of vertices are adjacent if and only
if the Hamming distance between the sequences that they represent is not larger than $\tau$. The size of $G$ (i.e.,
its number of edges) is equal to the number of pairs of sequences in $\set{A}$ which may not be distinguishable
by the decoder.

Further motivation for studying this problem is considered in the continuation (Section~\ref{subsection: influence}).

\subsection{Analysis}
\label{subsection: problem - analysis}
We next derive an upper bound on the size of the graph $G$.
Let $X^n=(X_1, \ldots, X_n)$ be chosen uniformly at random from the set $\set{A} \subseteq \{-1, 1\}^n$,
and let $\pmfOf{X^n}$ be the PMF of $X^n$. Then,

\vspace*{-0.3cm}
\begin{align}  \label{eq: 22032022a1}
\pmfOf{X^n}(x^n) =
\begin{dcases}
\frac1{\card{\set{A}}},  & \quad \text{if} \; x^n \in \set{A}, \\[0.1cm]
0, & \quad \text{if} \; x^n \not\in \set{A},
\end{dcases}
\end{align}
which implies that
\begin{eqnarray}  \label{eq: 22032022a2}
\Ent{X^n} = \log \card{\set{A}}.
\end{eqnarray}

The graph $G$ is an un-directed and simple graph with a vertex set $\V{G} = \set{A}$ (i.e., the vertices of $G$
are in one-to-one correspondence with the binary vectors in the set $\set{A}$). Its set of edges $\E{G}$
are the edges which connect all pairs of vertices in $G$ whose Hamming distance is less than or equal to $\tau$.
For $d \in \OneTo{\tau}$, let $\Ed{d}{G}$ be the set of edges in $G$ which connect all pairs of vertices
in $G$ whose Hamming distance is equal to $d$, so

\vspace*{-0.3cm}
\begin{equation}  \label{eq: 22032022a3}
\card{\E{G}} = \sum_{d=1}^{\tau} \card{\Ed{d}{G}}.
\end{equation}

For $x^n \in \{-1, 1\}^n$, $d \in \OneTo{n}$, and integers $k_1, \ldots, k_d$ such that $1 \leq k_1 < \ldots < k_d \leq n$, let
\begin{eqnarray} \label{eq: 22032022a4}
\widetilde{x}^{(k_1, \ldots, k_d)} \eqdef (x_1, \ldots, x_{k_1-1}, x_{k_1+1}, \ldots, x_{k_d-1}, x_{k_d+1}, \ldots, x_n)
\end{eqnarray}
be a subvector of $x^n$ of length $n-d$, obtained by dropping the bits of $x^n$ in positions $k_1, \ldots, k_d$; if $d=n$,
then $(k_1, \ldots, k_n) = (1, \ldots, n)$, and $\widetilde{x}^{(k_1, \ldots, k_d)}$ is an empty vector.
By the chain rule for the Shannon entropy,

\vspace*{-0.3cm}
\begin{subequations}   \label{eq: 220415a1}
\begin{align}  \nonumber
& \hspace{-0.5cm} \Ent{X^n} - \bigEnt{\widetilde{X}^{{(k_1, \ldots, k_d)}}} \\[0.1cm]
\label{eq: 22032022a6}
&= \bigEntCond{X_{k_1}, \ldots, X_{k_d} \,}{\, \widetilde{X}^{{(k_1, \ldots, k_d)}}} \\
\label{eq: 22032022a7}
&= -\sum_{x^n \in \{-1, 1\}^n} \pmfOf{X^n}(x^n) \, \log \Bigl( \CondpmfOf{X_{k_1}, \ldots,
X_{k_d} \, }{\, \widetilde{X}^{{(k_1, \ldots, k_d)}}}\bigl(x_{k_1}, \ldots, x_{k_d}
\, | \, \widetilde{x}^{(k_1, \ldots, k_d)}\bigr) \Bigr) \\
\label{eq: 22032022a8}
&= -\frac1{\card{\set{A}}} \sum_{x^n \in \set{A}} \log \Bigl( \CondpmfOf{X_{k_1}, \ldots,
X_{k_d} \, }{ \, \widetilde{X}^{{(k_1, \ldots, k_d)}}}\bigl(x_{k_1}, \ldots, x_{k_d}
\, | \, \widetilde{x}^{(k_1, \ldots, k_d)}\bigr) \Bigr),
\end{align}
\end{subequations}
where equality \eqref{eq: 22032022a8} holds by \eqref{eq: 22032022a1}.

For $x^n \in \{-1,1\}^n$, $d \in \OneTo{n}$, and integers $k_1, \ldots, k_d$ such that $1 \leq k_1 < \ldots < k_d \leq n$, let
\begin{eqnarray} \label{eq: 23032022a1}
\overline{x}^{(k_1, \ldots, k_d)} \eqdef (x_1, \ldots, x_{k_1-1}, -x_{k_1}, x_{k_1+1}, \ldots, x_{k_d-1}, -x_{k_d}, x_{k_d+1}, \ldots, x_n),
\end{eqnarray}
where the bits of $x^n$ in position $k_1, \ldots, k_d$ are flipped (in contrast to $\widetilde{x}^{(k_1, \ldots, k_d)}$ where the
bits of $x^n$ in these positions are dropped), so $\overline{x}^{(k_1, \ldots, k_d)} \in \{-1,1\}^n$ and $\bigdH{x^n}{\overline{x}^{(k_1, \ldots, k_d)}} = d$.
Likewise, if $x^n, y^n \in \{-1,1\}^n$ satisfy $\dH{x^n}{y^n} = d$, then there exist integers $k_1, \ldots, k_d$ such that
$1 \leq k_1 < \ldots < k_d \leq n$ where $y^n = \overline{x}^{(k_1, \ldots, k_d)}$ (i.e., the integers $k_1, \ldots, k_d$ are the positions
(in increasing order) where the vectors $x^n$ and $y^n$ differ).

Let us characterize the set $\set{A}$ by its cardinality, and the following two natural numbers:
\begin{enumerate}[a)]
\item If $x^n \in \set{A}$ and $\overline{x}^{(k_1, \ldots, k_d)} \in \set{A}$ for any $(k_1, \ldots, k_d)$ such that
$1 \leq k_1 < \ldots < k_d \leq n$, then there are at least $m_d \eqdef m_d(\set{A})$ vectors $y \in \set{A}$ whose
subvectors $\widetilde{y}^{(k_1, \ldots, k_d)}$ coincide with $\widetilde{x}^{(k_1, \ldots, k_d)}$, i.e., the integer
$m_d \geq 2$ satisfies
\begin{eqnarray}  \label{eq: 23032022a2}
\hspace*{-0.2cm} m_d \leq \min_{\substack{x^n \in \set{A}, \\[0.05cm] 1 \leq k_1 < \ldots < k_d \leq n}} \biggcard{\Bigl\{ y^n \in \set{A}:
\; \widetilde{y}^{(k_1, \ldots, k_d)} = \widetilde{x}^{(k_1, \ldots, k_d)}, \quad \overline{x}^{(k_1, \ldots, k_d)} \in \set{A} \Bigr\}}.
\end{eqnarray}
By definition, the integer $m_d$ always exists, and
\begin{eqnarray}  \label{eq: 23032022a3}
2 \leq m_d \leq \min\{2^d, \, \card{A} \}.
\end{eqnarray}
If no information is available about the value of $m_d$, then it can be taken by default to be equal to~2 (since by assumption
the two vectors $x^n \in \set{A}$ and $y^n \eqdef \overline{x}^{(k_1, \ldots, k_d)} \in \set{A}$ satisfy the equality
$\widetilde{y}^{(k_1, \ldots, k_d)} = \widetilde{x}^{(k_1, \ldots, k_d)}$).

\item
If $x^n \in \set{A}$ and $\overline{x}^{(k_1, \ldots, k_d)} \not\in \set{A}$ for any $(k_1, \ldots, k_d)$ such that
$1 \leq k_1 < \ldots < k_d \leq n$, then there are at least $\ell_d \eqdef \ell_d(\set{A})$ vectors $y^n \in \set{A}$ whose
subvectors $\widetilde{y}^{(k_1, \ldots, k_d)}$ coincide with $\widetilde{x}^{(k_1, \ldots, k_d)}$, i.e., the integer
$\ell_d \geq 1$ satisfies
\begin{eqnarray}  \label{eq: 23032022a4}
\hspace*{-0.2cm} \ell_d \leq \min_{\substack{x^n \in \set{A}, \\[0.05cm] 1 \leq k_1 < \ldots < k_d \leq n}} \biggcard{\Bigl\{ y^n \in \set{A}:
\; \widetilde{y}^{(k_1, \ldots, k_d)} = \widetilde{x}^{(k_1, \ldots, k_d)},
\quad \overline{x}^{(k_1, \ldots, k_d)} \not\in \set{A} \Bigr\}}.
\end{eqnarray}
By definition, the integer $\ell_d$ always exists, and
\begin{eqnarray}  \label{eq: 23032022a5}
1 \leq \ell_d \leq \min\{2^d-1, \, \card{A}-1 \}.
\end{eqnarray}
Likewise, if no information is available about the value of $\ell_d$, then it can be taken by default to be equal to~1 (since
$x^n \in \set{A}$ satisfies the requirement about its subvector $\widetilde{x}^{(k_1, \ldots, k_d)}$ in \eqref{eq: 23032022a4}).
\end{enumerate}
In general, it would be preferable to have the largest possible values of $m_d$ and $\ell_d$ (i.e., those satisfying
inequalities \eqref{eq: 23032022a2} and \eqref{eq: 23032022a4} with equalities, for obtaining a better upper bound
on the size of $G$ (this point will be clarified in the sequel). If $d=1$, then $m_d=2$ and $\ell_d=1$ are the best
possible constants (this holds by the definitions in \eqref{eq: 23032022a2} and \eqref{eq: 23032022a4}, which can be
also verified by the coincidence of the upper and lower bounds in \eqref{eq: 23032022a3} for $d=1$, as well as
those in \eqref{eq: 23032022a5}).

If $x^n \in \set{A}$, then we distinguish between the following two cases:
\begin{itemize}
\item If $\overline{x}^{(k_1, \ldots, k_d)} \in \set{A}$, then
\begin{eqnarray}  \label{eq: 23032022a6}
\CondpmfOf{X_{k_1}, \ldots, X_{k_d} \,}{\, \widetilde{X}^{(k_1, \ldots, k_d)}}(x_{k_1}, \ldots, x_{k_d} \, |
\, \widetilde{x}^{(k_1, \ldots, k_d)}) \leq \frac1{m_d},
\end{eqnarray}
which holds by the way that $m_d$ is defined in \eqref{eq: 23032022a2}, and since $X^n$ is randomly selected to
be equiprobable in the set $\set{A}$.
\item
If $\overline{x}^{(k_1, \ldots, k_d)} \not\in \set{A}$, then
\begin{eqnarray}  \label{eq: 23032022a7}
\CondpmfOf{X_{k_1}, \ldots, X_{k_d} \,}{\, \widetilde{X}^{(k_1, \ldots, k_d)}}(x_{k_1}, \ldots, x_{k_d} \, |
\, \widetilde{x}^{(k_1, \ldots, k_d)}) \leq \frac1{\ell_d},
\end{eqnarray}
which holds by the way that $\ell_d$ is defined in \eqref{eq: 23032022a4}, and since $X^n$ is equiprobable on $\set{A}$.
\end{itemize}

For $d \in \OneTo{\tau}$ and $1 \leq k_1 < \ldots < k_d \leq n$, it follows from
\eqref{eq: 220415a1}, \eqref{eq: 23032022a6} and \eqref{eq: 23032022a7} that

\vspace*{-0.3cm}
\begin{align}
\Ent{X^n} - \bigEnt{\widetilde{X}^{{(k_1, \ldots, k_d)}}}
& \geq \frac{\log m_d}{\card{A}} \; \sum_{x^n} \I{x^n \in \set{A}, \; \overline{x}^{(k_1, \ldots, k_d)} \in \set{A} }  \nonumber \\[0.05cm]
\label{eq: 23032022a8}
&\hspace*{0.4cm} + \frac{\log \ell_d}{\card{A}} \; \sum_{x^n} \I{x^n \in \set{A}, \; \overline{x}^{(k_1, \ldots, k_d)} \not\in \set{A} },
\end{align}
which, by summing on both sides of inequality \eqref{eq: 23032022a8} over all the integers $k_1, \ldots, k_d$
such that $1 \leq k_1 < \ldots < k_d \leq n$, yields

\vspace*{-0.3cm}
\begin{align}
& \hspace{-0.5cm} \sum_{\substack{(k_1, \ldots, k_d): \\[0.05cm] 1 \leq k_1 < \ldots < k_d \leq n}} \Bigl( \Ent{X^n} - \bigEnt{\widetilde{X}^{{(k_1, \ldots, k_d)}}} \Bigr) \nonumber \\
& \geq \frac{\log m_d}{\card{A}} \; \sum_{\substack{(k_1, \ldots, k_d): \\[0.05cm] 1 \leq k_1 < \ldots < k_d \leq n}}
\sum_{x^n} \I{x^n \in \set{A}, \; \overline{x}^{(k_1, \ldots, k_d)} \in \set{A} }  \nonumber \\[0.05cm]
\label{eq: 23032022a9}
& \hspace*{0.4cm} + \frac{\log \ell_d}{\card{A}} \; \sum_{\substack{(k_1, \ldots, k_d): \\[0.05cm] 1 \leq k_1 < \ldots < k_d \leq n}} \sum_{x^n}
\I{x^n \in \set{A}, \; \overline{x}^{(k_1, \ldots, k_d)} \not\in \set{A} }.
\end{align}
Equality holds in \eqref{eq: 23032022a9} if the minima on the RHS of \eqref{eq: 23032022a2} and \eqref{eq: 23032022a4}
are attained by any element in these sets, and if \eqref{eq: 23032022a2} and
\eqref{eq: 23032022a4} are satisfied with equalities (i.e., $m_d$ and $\ell_d$ are the maximal integers to satisfy inequalities
\eqref{eq: 23032022a2} and \eqref{eq: 23032022a4} for the given set $\set{A}$). Hence, this equality holds in particular for $d=1$,
with the constants $m_d = 2$ and $\ell_d = 1$.

The double sum in the first term on the RHS of \eqref{eq: 23032022a9} is equal to
\begin{eqnarray}  \label{eq: 23032022a10}
\sum_{\substack{(k_1, \ldots, k_d): \\[0.05cm] 1 \leq k_1 < \ldots < k_d \leq n}}
\sum_{x^n} \I{x^n \in \set{A}, \; \overline{x}^{(k_1, \ldots, k_d)} \in \set{A} } = 2 \, \bigcard{\Ed{d}{G}},
\end{eqnarray}
since every pair of adjacent vertices in $\set{G}$ that refer to vectors in $\set{A}$ whose Hamming distance is equal to $d$
is of the form $x^n \in \set{A}$ and $\overline{x}^{(k_1, \ldots, k_d)} \in \set{A}$, and vice versa, and
every edge $\{x^n, \overline{x}^{(k_1, \ldots, k_d)} \} \in \Ed{d}{G}$ is counted twice in the double summation
on the LHS of \eqref{eq: 23032022a10}. For calculating the double sum in the second term
on the RHS of \eqref{eq: 23032022a9}, we first calculate the sum of these two double summations:
\begin{subequations}
\begin{eqnarray}
&& \hspace*{-0.5cm} \sum_{\substack{(k_1, \ldots, k_d): \\[0.05cm] 1 \leq k_1 < \ldots < k_d \leq n}}
\sum_{x^n} \I{x^n \in \set{A}, \; \overline{x}^{(k_1, \ldots, k_d)} \in \set{A} } +
\sum_{\substack{(k_1, \ldots, k_d): \\[0.05cm] 1 \leq k_1 < \ldots < k_d \leq n}} \sum_{x^n}
\I{x^n \in \set{A}, \; \overline{x}^{(k_1, \ldots, k_d)} \not\in \set{A} }  \nonumber \\
\label{eq: 23032022a11}
&& = \sum_{\substack{(k_1, \ldots, k_d): \\[0.05cm] 1 \leq k_1 < \ldots < k_d \leq n}}
\sum_{x^n} \Bigl\{ \I{x^n \in \set{A}, \; \overline{x}^{(k_1, \ldots, k_d)} \in \set{A} } +
\I{x^n \in \set{A}, \; \overline{x}^{(k_1, \ldots, k_d)} \not\in \set{A} } \Bigr\} \\
\label{eq: 23032022a12}
&& = \sum_{\substack{(k_1, \ldots, k_d): \\[0.05cm] 1 \leq k_1 < \ldots < k_d \leq n}}
\sum_{x^n} \I{x^n \in \set{A}} \\
\label{eq: 23032022a13}
&& = \sum_{\substack{(k_1, \ldots, k_d): \\[0.05cm] 1 \leq k_1 < \ldots < k_d \leq n}} \card{\set{A}} \\
\label{eq: 23032022a14}
&& = \binom{n}{d} \, \card{\set{A}},
\end{eqnarray}
\end{subequations}
so, subtracting \eqref{eq: 23032022a10} from \eqref{eq: 23032022a14} gives that
\begin{eqnarray}  \label{eq: 23032022a15}
\sum_{\substack{(k_1, \ldots, k_d): \\[0.05cm] 1 \leq k_1 < \ldots < k_d \leq n}} \sum_{x^n}
\I{x^n \in \set{A}, \; \overline{x}^{(k_1, \ldots, k_d)} \not\in \set{A} } =
\binom{n}{d} \, \card{\set{A}} - 2 \, \bigcard{\Ed{d}{G}}.
\end{eqnarray}
Substituting \eqref{eq: 23032022a10} and \eqref{eq: 23032022a15} into the RHS of \eqref{eq: 23032022a9}
gives that, for all $d \in \OneTo{\tau}$,
\begin{subequations}   \label{eq: 220417a1}
\begin{eqnarray}  \nonumber
&& \hspace*{-1cm} \sum_{\substack{(k_1, \ldots, k_d): \\[0.05cm] 1 \leq k_1 < \ldots < k_d \leq n}}
\Bigl( \Ent{X^n} - \bigEnt{\widetilde{X}^{{(k_1, \ldots, k_d)}}} \Bigr) \\[0.1cm]
\label{eq: 23032022a16}
&& \geq \frac{2 \card{\Ed{d}{G}} \, \log m_d}{\card{A}} + \frac{\log \ell_d}{\card{\set{A}}}
\, \biggl[ \binom{n}{d} \, \card{\set{A}} - 2 \, \bigcard{\Ed{d}{G}} \biggr] \\[0.1cm]
\label{eq: 23032022a17}
&& = \binom{n}{d} \, \log \ell_d + \frac{2 \card{\Ed{d}{G}}}{\card{\set{A}}} \, \log \frac{m_d}{\ell_d},
\end{eqnarray}
\end{subequations}
with the same necessary and sufficient condition for equality in \eqref{eq: 23032022a16} as in \eqref{eq: 23032022a9}.
(Recall that it is in particular an equality for $d=1$, where in this case $m_1 = 2$ and $\ell_1 = 1$.)

By the generalized Han's inequality in \eqref{eq1: Han's inequality},

\vspace*{-0.3cm}
\begin{subequations}    \label{eq: 220417a2}
\begin{align}  \label{eq: 23032022a18}
\hspace*{-0.5cm} \sum_{\substack{(k_1, \ldots, k_d): \\[0.05cm] 1 \leq k_1 < \ldots < k_d \leq n}}
\Bigl( \Ent{X^n} - \bigEnt{\widetilde{X}^{{(k_1, \ldots, k_d)}}} \Bigr) & \leq \binom{n-1}{d-1} \, \Ent{X^n} \\[-0.2cm]
&= \binom{n-1}{d-1} \, \log \card{\set{A}},  \label{eq: 23032022a19}
\end{align}
\end{subequations}
where equality \eqref{eq: 23032022a19} holds by \eqref{eq: 22032022a2}. Combining \eqref{eq: 220417a1} and
\eqref{eq: 220417a2} yields
\begin{eqnarray}  \label{eq: 23032022a20}
\binom{n-1}{d-1} \, \log \card{\set{A}} \geq \binom{n}{d} \, \log \ell_d + \frac{2 \card{\Ed{d}{G}}}{\card{\set{A}}} \, \log \frac{m_d}{\ell_d},
\end{eqnarray}
and, by the identity $\binom{n}{d} = \frac{n}{d} \, \binom{n-1}{d-1}$, we get
\begin{eqnarray}  \label{eq: 23032022a21}
\card{\Ed{d}{G}} \leq \frac{\binom{n-1}{d-1} \, \card{\set{A}} \,
\Bigl( \log \card{\set{A}} - \frac{n}{d} \, \log \ell_d \Bigr)}{2 \log \frac{m_d}{\ell_d}}.
\end{eqnarray}
This upper bound is specialized, for $d=1$, to Theorem~4.2 of \cite{Boucheron_Lugosi_Massart_book}
(where, by definition, $m_1 = 2$ and $\ell_1=1$). This gives that the
number of edges in $G$, connecting pairs of vertices which refer to binary
vectors in $\set{A}$ whose Hamming distance is~1 from each other, satisfies
\begin{eqnarray}  \label{eq: 23032022a22}
\card{\Ed{1}{G}} \leq \tfrac12 \, \card{\set{A}} \, \log_2 \card{\set{A}}.
\end{eqnarray}
It is possible to select, by default, the values of the integers $m_d$ and $\ell_d$ to be equal
to~2 and 1, respectively, independently of the value of $d \in \OneTo{\tau}$. It therefore follows
that the upper bound in \eqref{eq: 23032022a21} can be loosened to
\begin{eqnarray}  \label{eq: 23032022a23}
\card{\Ed{d}{G}} \leq \tfrac12 \, \binom{n-1}{d-1} \, \card{\set{A}} \, \log_2 \card{\set{A}}.
\end{eqnarray}
This shows that the bound in \eqref{eq: 23032022a23} generalizes the result in Theorem~4.2 of \cite{Boucheron_Lugosi_Massart_book},
based only on the knowledge of the cardinality of $\set{A}$. Furthermore, the bound \eqref{eq: 23032022a23} can be tightened by the refined
bound \eqref{eq: 23032022a21} if the characterization of the set $\set{A}$ allows one to assert values for $m_d$ and $\ell_d$
that are larger than the trivial values of~2 and~1, respectively.

In light of \eqref{eq: 22032022a3} and \eqref{eq: 23032022a23}, the number of edges in the graph $G$ satisfies
\begin{eqnarray}  \label{eq: 23032022a24}
\card{\E{G}} \leq \tfrac12 \sum_{d=1}^{\tau} \binom{n-1}{d-1} \, \card{\set{A}} \, \log_2 \card{\set{A}},
\end{eqnarray}
and if $\tau \leq \frac{n+1}{2}$, then it follows that
\begin{eqnarray}  \label{eq: 23032022a25}
\card{\E{G}} \leq \tfrac12 \exp\biggl( (n-1) \, \biggEntBin{\frac{\tau-1}{n-1}} \biggr) \, \card{\set{A}} \, \log_2 \card{\set{A}}.
\end{eqnarray}
Indeed, the transition from \eqref{eq: 23032022a24} to \eqref{eq: 23032022a25} holds by the inequality
\begin{eqnarray}  \label{eq: 23032022a26}
&& \sum_{k=0}^{n \theta} \binom{n}{k} \leq \exp \bigl( n \EntBin{\theta} \bigr), \quad \theta \in \bigl[0, \tfrac12 \bigr],
\end{eqnarray}
where the latter bound is asymptotically tight in the exponent of $n$ (for sufficiently large values of $n$).

\subsection{Comparison of Bounds}
\label{subsection: comparison of bounds}
We next consider the tightness of the refined bound \eqref{eq: 23032022a21} and the loosened bound
\eqref{eq: 23032022a23}.
Since $\set{A}$ is a subset of the $n$-dimensional cube $\{-1,1\}^n$, every point in $\set{A}$ has at most $\binom{n}{d}$
neighbors in $\set{A}$ with Hamming distance $d$, so
\begin{eqnarray}  \label{eq: 24032022a4}
\card{\Ed{d}{G}} \leq \tfrac12 \, \binom{n}{d} \, \card{\set{A}}.
\end{eqnarray}
Comparing the bound on the RHS of \eqref{eq: 23032022a21} with the trivial bound in \eqref{eq: 24032022a4}
shows that the former bound is useful if and only if
\begin{eqnarray} \label{eq: 24032022a5}
\frac{\log \card{\set{A}} - \frac{n}{d} \, \log \ell_d}{\log \frac{m_d}{\ell_d}} \leq \frac{n}{d},
\end{eqnarray}
which is obtained by relying on the identity $\binom{n}{d} = \frac{n}{d} \, \binom{n-1}{d-1}$. Rearranging terms
in \eqref{eq: 24032022a5} gives the necessary and sufficient condition
\begin{eqnarray}  \label{eq: 24032022a6}
\card{\set{A}} \leq (m_d)^{\frac{n}{d}},
\end{eqnarray}
which is independent of the value of $\ell_d$.
Since, by definition, $m_d \geq 2$, inequality \eqref{eq: 24032022a6} is automatically satisfied if the stronger condition
\begin{eqnarray}  \label{eq: 24032022a7}
\card{\set{A}} \leq 2^{\frac{n}{d}}
\end{eqnarray}
is imposed.
The latter also forms a necessary and sufficient condition for the usefulness of the looser bound on the
RHS of \eqref{eq: 23032022a23} in comparison to \eqref{eq: 24032022a4}.
\begin{example} \label{example: tightness of bounds}
Suppose that the set $\set{A} \subseteq \{-1, 1\}^n$ is characterized by the property that for all $d \in \OneTo{\tau}$,
with a fixed integer $\tau \in \OneTo{n}$, if $x^n \in \set{A}$ and $\overline{x}^{(k_1, \ldots, k_d)} \in \set{A}$ then all
vectors $y^n \in \{-1,1\}^n$ which coincide with $x^n$ and $\overline{x}^{(k_1, \ldots, k_d)}$ in their $(n-d)$ agreed
positions are also included in the set $\set{A}$. Then, for all $d \in \OneTo{\tau}$, we get by definition that $m_d = 2^d$, which
yields $\tau \leq \lfloor \log_2 \card{\set{A}} \rfloor$. Setting $m_d = 2^d$ and the default value $\ell_d = 1$ on the RHS of
\eqref{eq: 23032022a21} gives

\vspace*{-0.3cm}
\begin{subequations}
\begin{align}
\label{eq: 24032022a8}
\card{\Ed{d}{G}} &\leq \frac{\binom{n-1}{d-1} \, \card{\set{A}} \,
\Bigl( \log \card{\set{A}} - \frac{n}{d} \, \log \ell_d \Bigr)}{2 \log \frac{m_d}{\ell_d}} \\
\label{eq: 24032022a9}
&= \frac{\binom{n-1}{d-1} \, \card{\set{A}} \, \log \card{\set{A}}}{2 \log (2^d)} \\[0.1cm]
\label{eq: 24032022a10}
&= \frac1{2d} \, \binom{n-1}{d-1} \, \card{\set{A}} \, \log_2 \card{\set{A}} \\[0.1cm]
\label{eq: 24032022a12}
&= \frac12 \, \binom{n}{d} \, \card{\set{A}} \cdot \frac{\log_2 \card{\set{A}}}{n}.
\end{align}
\end{subequations}
Unless $\set{A} = \{-1, 1\}^n$, the upper bound on the RHS of \eqref{eq: 24032022a12}
is strictly smaller than the trivial upper bound on the RHS of \eqref{eq: 24032022a4}.
This improvement is consistent with the satisfiability of the (necessary and sufficient) condition
in \eqref{eq: 24032022a7}, which is strictly satisfied since
\begin{eqnarray}  \label{eq: 24032022a13}
\card{\set{A}} < 2^n = (2^d)^{\frac{n}{d}} = (m_d)^{\frac{n}{d}}.
\end{eqnarray}
On the other hand, the looser upper bound on the RHS of \eqref{eq: 23032022a23} gives
\begin{eqnarray}
\label{eq: 24032022a14}
\card{\Ed{d}{G}} \leq \frac12 \, \binom{n}{d} \, \card{\set{A}} \cdot \frac{d \, \log_2 \card{\set{A}}}{n},
\end{eqnarray}
which is $d$ times larger than the refined bound on the RHS of \eqref{eq: 24032022a12} (since it is based on the
exact value of $m_d$ for the set $\set{A}$, rather than taking the default value of~2), and it is worse than the
trivial bound if and only if $\card{\set{A}} > 2^{\frac{n}{d}}$.
The latter finding is consistent with \eqref{eq: 24032022a7}.

This exemplifies the utility of the refined upper bound on the RHS of \eqref{eq: 23032022a21}
in comparison to the bound on the RHS of \eqref{eq: 23032022a23}, where the latter generalizes
\cite[Theorem~4.2]{Boucheron_Lugosi_Massart_book} from the case where $d=1$ to all $d \in \OneTo{n}$. As
it is explained above, this refinement is irrelevant in the special case where $d=1$, though it proves to
be useful in general for $d \in \FromTo{2}{n}$ (as it is exemplified here).
\end{example}

The following theorem introduces the results of our analysis (so far) in the present section.
\begin{theorem}  \label{theorem: graphs}
Let $\set{A} \subseteq \{-1, 1\}^n$, with $n \in \naturals$, and let $\tau \in \OneTo{n}$.
Let $G = (\V{G}, \E{G})$ be an un-directed, simple graph with vertex set $\V{G} = \set{A}$, and edges connecting pairs of
vertices in $G$ which are represented by vectors in $\set{A}$ whose Hamming distance is less than or equal to $\tau$.
For $d \in \OneTo{\tau}$, let $\Ed{d}{G}$ be the set of edges in $G$ which connect all pairs of vertices
that are represented by vectors in $\set{A}$ whose Hamming distance is equal to $d$
(i.e., $\card{\E{G}} = \sum_{d=1}^{\tau} \card{\Ed{d}{G}}$).
\begin{enumerate}[a)]
\item For $d \in \OneTo{\tau}$, let the integers $m_d \in \FromTo{2}{\min\{2^d, \, \card{A} \}}$ and
$\ell_d  \in \OneTo{\min\{2^d-1, \, \card{A}-1 \}}$ (be, preferably, the maximal possible values to)
satisfy the requirements in \eqref{eq: 23032022a2} and \eqref{eq: 23032022a4}, respectively. Then,
\begin{eqnarray}  \label{eq: 24032022a1}
\card{\Ed{d}{G}} \leq \frac{\binom{n-1}{d-1} \, \card{\set{A}} \,
\Bigl( \log \card{\set{A}} - \frac{n}{d} \, \log \ell_d \Bigr)}{2 \log \frac{m_d}{\ell_d}}.
\end{eqnarray}
\item A loosened bound, which only depends on the cardinality of the set $\set{A}$, is obtained
by setting the default values $m_d=2$ and $\ell_d=1$. It is then given by
\begin{eqnarray}  \label{eq: 24032022a2}
\card{\Ed{d}{G}} \leq \tfrac12 \, \binom{n-1}{d-1} \, \card{\set{A}} \, \log_2 \card{\set{A}},
\quad d \in \OneTo{\tau},
\end{eqnarray}
and, if $\tau \leq \frac{n+1}{2}$, then the (overall) number of edges in $G$ satisfies
\begin{eqnarray}  \label{eq: 24032022a3}
\card{\E{G}} \leq \tfrac12 \exp\biggl( (n-1) \, \biggEntBin{\frac{\tau-1}{n-1}} \biggr) \, \card{\set{A}} \, \log_2 \card{\set{A}}.
\end{eqnarray}
\item The refined upper bound on the RHS of \eqref{eq: 24032022a1} and the loosened upper bound on the RHS of
\eqref{eq: 24032022a2} improve the trivial bound $\tfrac12 \binom{n}{d} \card{\set{A}}$,
if and only if $\card{\set{A}} < (m_d)^{\frac{n}{d}}$ or $\card{\set{A}} < 2^{\frac{n}{d}}$, respectively (see
Example~\ref{example: tightness of bounds}).
\end{enumerate}
\end{theorem}

\subsection{Influence of Fixed-Size Subsets of Bits}
\label{subsection: influence}

The result in Theorem~4.2 of \cite{Boucheron_Lugosi_Massart_book}, which is generalized and refined in
Theorem~\ref{theorem: graphs} here, is turned to study the total influence of the $n$ variables
of an equiprobable random vector $X^n \in \{-1,1\}^n$ on a subset $\set{A} \subset \{-1,1\}^n$.
To this end, let $\overline{X}^{(i)}$ denote the vector where the bit at the $i$-th position
of $X^n$ is flipped, so $\overline{X}^{(i)} \eqdef (X_1, \ldots, X_{i-1}, -X_i, X_{i+1}, \ldots, X_n)$
for all $i \in \OneTo{n}$. Then, the influence of the $i$-th variable is defined as
\begin{eqnarray} \label{eq: 25032022a1}
I_{i}(\set{A}) \eqdef \BigPrv{ \I{X^n \in \set{A}} \neq \bigI{\overline{X}^{(i)} \in \set{A}} },  \quad i \in \OneTo{n},
\end{eqnarray}
and their total influence is defined to be the sum
\begin{eqnarray}  \label{eq: 25032022a2}
I(\set{A}) \eqdef \sum_{i=1}^n I_{i}(\set{A}).
\end{eqnarray}
As it is shown in Chapters~9 and~10 of \cite{Boucheron_Lugosi_Massart_book}, influences of subsets
of the binary hypercube have far reaching consequences in the study of threshold phenomena,
and many other areas. As a corollary of \eqref{eq: 23032022a22}, it is obtained in
Theorem~4.3 of \cite{Boucheron_Lugosi_Massart_book} that, for every subset $\set{A} \subset \{-1,1\}^n$,
\begin{eqnarray}  \label{eq: 25032022a3}
I(\set{A}) \geq 2 \Prs{\set{A}} \, \log_2 \frac1{\Prs{\set{A}}},
\end{eqnarray}
where $\Prs{\set{A}} \eqdef \prob[X^n \in \set{A}] = \frac{\card{\set{A}}}{2^n}$ by the equiprobable
distribution of $X^n$ over $\{-1,1\}^n$.

In light of Theorem~\ref{theorem: graphs}, the same approach which is used in Section~4.4 of \cite{Boucheron_Lugosi_Massart_book}
for the transition from \eqref{eq: 23032022a22} to \eqref{eq: 25032022a3} can be also used to obtain, as a corollary, a lower
bound on the average total influence over all subsets of $d$ variables. To this end, let $k_1, \ldots, k_d$ be integers
such that $1 \leq k_1 < \ldots < k_d \leq n$, and the influence of the variables in positions $k_1, \ldots, k_d$ be given by
\begin{eqnarray} \label{eq: 25032022a4}
I_{(k_1, \ldots, k_d)}(\set{A}) \eqdef \BigPrv{ \I{X^n \in \set{A}} \neq \BigI{\overline{X}^{(k_1, \ldots, k_d)} \in \set{A}} }.
\end{eqnarray}
Then, let the average influence of subsets of $d$ variables be defined as
\begin{eqnarray} \label{eq: 25032022a5}
I^{(n,d)}(\set{A}) \eqdef \frac1{\binom{n}{d}} \sum_{\substack{(k_1, \ldots, k_d): \\[0.05cm] 1 \leq k_1 < \ldots < k_d \leq n}}
I_{(k_1, \ldots, k_d)}(\set{A}).
\end{eqnarray}
Hence, by \eqref{eq: 25032022a2} and \eqref{eq: 25032022a5}, $ I^{(n,1)}(\set{A}) = \frac1n \, I(\set{A})$ for every subset $\set{A} \subset \{-1, 1\}^n$.
Let
\begin{eqnarray}  \label{eq: 25032022a6}
\set{B}^{(n,d)}(\set{A}) \eqdef \bigl\{ (x^n, y^n): \, x^n \in \set{A}, \quad y^n \in \{-1, 1\}^n \setminus \set{A}, \quad \dH{x^n}{y^n} = d \bigr\},
\end{eqnarray}
be the set of ordered pairs of sequences $(x^n , y^n)$, where $x^n ,y^n \in \{-1, 1\}^n$ are of Hamming distance $d$ from each other,
with $x^n \in \set{A}$ and $y^n \not\in \set{A}$.
By the equiprobable distribution of $X^n$ on $\{-1,1\}^n$, we get

\vspace*{-0.3cm}
\begin{subequations}  \label{eq: 25032022a7}
\begin{align}
I^{(n,d)}(\set{A}) &= \frac1{\binom{n}{d}} \sum_{\substack{(k_1, \ldots, k_d): \\[0.05cm] 1 \leq k_1 < \ldots < k_d \leq n}}
\BigPrv{ \I{X^n \in \set{A}} \neq \BigI{\overline{X}^{(k_1, \ldots, k_d)} \in \set{A}} } \label{eq1: 25032022a7} \\
&= \frac2{\binom{n}{d}} \sum_{\substack{(k_1, \ldots, k_d): \\[0.05cm] 1 \leq k_1 < \ldots < k_d \leq n}}
\BigPrv{ X^n \in \set{A}, \; \; \overline{X}^{(k_1, \ldots, k_d)} \not\in \set{A} } \label{eq2: 25032022a7} \\
&= \frac{2}{\binom{n}{d}} \cdot \frac{\card{\set{B}^{(n,d)}(\set{A})}}{2^n} \label{eq3: 25032022a7} \\[0.1cm]
&= \frac{\card{\set{B}^{(n,d)}(\set{A})}}{2^{n-1} \, \binom{n}{d}}.  \label{eq4: 25032022a7}
\end{align}
\end{subequations}
Since every point in $\set{A}$ has $\binom{n}{d}$ neighbors of Hamming distance $d$ in the set $\{-1, 1\}^n$,
it follows that
\begin{eqnarray}  \label{eq: 25032022a8}
\binom{n}{d} \, \card{\set{A}} = 2 \, \bigcard{\Ed{d}{G}} + \bigcard{\set{B}^{(n,d)}(\set{A})},
\end{eqnarray}
where $G$ is introduced in Theorem~\ref{theorem: graphs}, and $\Ed{d}{G}$ is the set of
edges connecting pairs of vertices in $G$ which are represented by vectors in $\set{A}$ of Hamming
distance $d$. The multiplication by~2 on the RHS of \eqref{eq: 25032022a8} is because every
edge whose two endpoints are in the set $\set{A}$ is counted twice. Hence, by \eqref{eq: 23032022a21} and \eqref{eq: 25032022a8},

\vspace*{-0.3cm}
\begin{subequations}
\begin{align}  \label{eq: 25032022a9}
\bigcard{\set{B}^{(n,d)}(\set{A})} &= \binom{n}{d} \, \card{\set{A}} - 2 \, \bigcard{\Ed{d}{G}} \\[-0.1cm]
&\geq \binom{n}{d} \, \card{\set{A}} - \frac{\binom{n-1}{d-1} \, \card{\set{A}} \,
\Bigl( \log \card{\set{A}} - \frac{n}{d} \, \log \ell_d \Bigr)}{\log \frac{m_d}{\ell_d}}  \label{eq: 25032022a10} \\
&= \binom{n}{d} \, \card{\set{A}} \, \biggl( 1 - \frac{\frac{d}{n} \log \card{\set{A}}
- \log \ell_d}{\log \frac{m_d}{\ell_d}} \biggr) \label{eq: 25032022a11} \\
&= \binom{n}{d} \, \card{\set{A}} \, \biggl( \frac{\log m_d  - \frac{d}{n} \, \log \card{\set{A}}}{\log \frac{m_d}{\ell_d}} \biggr), \label{eq: 25032022a12}
\end{align}
\end{subequations}
and the lower bound on the RHS of \eqref{eq: 25032022a12} is positive if and only if $\card{\set{A}} < (m_d)^{\frac{n}{d}}$ (see also \eqref{eq: 24032022a6}).
This gives from \eqref{eq: 25032022a7} that the average influence of subsets of $d$ variables satisfies

\vspace*{-0.3cm}
\begin{subequations}
\begin{align}
\label{eq: 25032022a13}
I^{(n,d)}(\set{A}) &\geq \frac{\card{\set{A}}}{2^{n-1}}
\, \biggl( \frac{\log m_d  - \frac{d}{n} \, \log \card{\set{A}}}{\log \frac{m_d}{\ell_d}} \biggr) \\
\label{eq: 25032022a14}
&= 2 \Prs{\set{A}} \, \biggl( \frac{\log m_d  - \frac{d}{n} \, \log \bigl( 2^n \Prs{\set{A}} \bigr)}{\log \frac{m_d}{\ell_d}} \biggr) \\
\label{eq: 25032022a15}
&= 2 \Prs{\set{A}} \, \Biggl( \frac{\frac{d}{n} \, \log \frac1{\Prs{\set{A}}} - \log \frac{2^d}{m_d}  }{\log \frac{m_d}{\ell_d}} \Biggr).
\end{align}
\end{subequations}
Note that by setting $d=1$, and the default values $m_d = 2$ and $\ell_d = 1$ on the RHS of \eqref{eq: 25032022a15} gives
the total influence of the $n$ variables satisfies, for all $\set{A} \subseteq \{-1, 1\}^n$,

\vspace*{-0.3cm}
\begin{subequations}
\begin{align}
\label{eq: 25032022a16}
I(\set{A}) &= n I^{(n,1)}(\set{A}) \\
& \geq 2 \Prs{\set{A}} \, \log_2 \frac1{\Prs{\set{A}}},
\end{align}
\end{subequations}
which is then specialized to the result in \cite[Theorem~4.3]{Boucheron_Lugosi_Massart_book} (see \eqref{eq: 25032022a3}).
This gives the following result.
\begin{theorem} \label{theorem: influence}
Let $X^n$ be an equiprobable random vector over the set $\{-1,1\}^n$,
let $d \in \OneTo{n}$ and $\set{A} \subset \{-1, 1\}^n$. Then, the average
influence of subsets of $d$ variables of $X^n$, as it is defined in \eqref{eq: 25032022a5},
is lower bounded as follows:

\vspace*{-0.3cm}
\begin{align}
\label{eq: 25032022a17}
I^{(n,d)}(\set{A}) &\geq 2 \Prs{\set{A}} \, \Biggl( \frac{\frac{d}{n}
\, \log \frac1{\Prs{\set{A}}} - \log \frac{2^d}{m_d}}{\log \frac{m_d}{\ell_d}} \Biggr),
\end{align}
where $\Prs{\set{A}} \eqdef \prob[X^n \in \set{A}] = \frac{\card{\set{A}}}{2^n}$,
and the integers $m_d$ and $\ell_d$ are introduced in Theorem~\ref{theorem: graphs}.
Similarly to the refined upper bound in Theorem~\ref{theorem: graphs}, the lower bound on the
RHS of \eqref{eq: 25032022a17} is informative (i.e., positive) if and only if
$\card{\set{A}} < (m_d)^{\frac{n}{d}}$. The lower bound on the RHS of
\eqref{eq: 25032022a17} can be loosened (by setting the default values $m_d=2$
and $\ell_d=1$) to

\vspace*{-0.3cm}
\begin{align}
\label{eq: 25032022a18}
I^{(n,d)}(\set{A}) &\geq 2 \Prs{\set{A}} \, \biggl( \frac{d}{n}
\, \log_2 \frac1{\Prs{\set{A}}} + 1-d \biggr).
\end{align}
\end{theorem}

\appendices

\setcounter{section}{0}

\section{Proof of Proposition~\ref{proposition: information measures, polymatroids}}
\label{appendix A: proof}

\setcounter{equation}{0}
\renewcommand{\theequation}{\thesection.\arabic{equation}}

For completeness, we prove Proposition~\ref{proposition: information measures, polymatroids}
which introduces results from \cite{Fujishige78}, \cite{Krause_UAI05} and \cite{Madiman_ITW08}.

Let $\Omega$ be a non-empty finite set, and let $\{X_\omega\}_{\omega \in \Omega}$ be a collection of
discrete random variables. We first prove Item~(a), showing that the entropy set function
$f \colon 2^{\Omega} \to \Reals$ in \eqref{entropic function} is a rank function.
\begin{itemize}
\item $f(\es)=0$.
\item Submodularity: If $\set{S}, \set{T} \subseteq \Omega$, then
\begin{subequations}  \label{eq1: submodularity - entropy}
\begin{eqnarray}
&& \hspace*{-1.0cm} f(\set{T} \cup \set{S}) + f(\set{T} \cap \set{S}) \nonumber \\[0.1cm]
&& \hspace*{-0.7cm} = \Ent{X_{\set{T} \cup \set{S}}} + \Ent{X_{\set{T} \cap \set{S}}} \\[0.1cm]
&& \hspace*{-0.7cm} = \Ent{X_{\set{T} \setminus \set{S}}, X_{\set{T} \cap \set{S}}, X_{\set{S} \setminus \set{T}}} + \Ent{X_{\set{T} \cap \set{S}}} \\[0.1cm]
&& \hspace*{-0.7cm} = \EntCond{X_{\set{T} \setminus \set{S}}, X_{\set{S} \setminus \set{T}}}{X_{\set{T} \cap \set{S}}} + 2 \Ent{X_{\set{T} \cap \set{S}}} \\[0.1cm]
&& \hspace*{-0.7cm} = \bigl[ \EntCond{X_{\set{T} \setminus \set{S}}}{X_{\set{T} \cap \set{S}}} + \EntCond{X_{\set{S} \setminus \set{T}}}{X_{\set{T} \cap \set{S}}}
-\MInfoCond{X_{\set{T} \setminus \set{S}}}{X_{\set{S} \setminus \set{T}}}{X_{\set{T} \cap \set{S}}} \bigr] \nonumber \\[0.1cm]
&& \hspace*{-0.2cm} + 2 \Ent{X_{\set{T} \cap \set{S}}} \\[0.1cm]
&& \hspace*{-0.7cm} = \bigl[ \EntCond{X_{\set{T} \setminus \set{S}}}{X_{\set{T} \cap \set{S}}} + \Ent{X_{\set{T} \cap \set{S}}} \bigr]
+ \bigl[ \EntCond{X_{\set{S} \setminus \set{T}}}{X_{\set{T} \cap \set{S}}} + \Ent{X_{\set{T} \cap \set{S}}} \bigr] \nonumber \\
&& \hspace*{-0.2cm} -\MInfoCond{X_{\set{T} \setminus \set{S}}}{X_{\set{S} \setminus \set{T}}}{X_{\set{T} \cap \set{S}}} \\[0.1cm]
&& \hspace*{-0.7cm} = \Ent{X_{\set{T} \setminus \set{S}}, X_{\set{T} \cap \set{S}}} + \Ent{X_{\set{S} \setminus \set{T}}, X_{\set{T} \cap \set{S}}}
-\MInfoCond{X_{\set{T} \setminus \set{S}}}{X_{\set{S} \setminus \set{T}}}{X_{\set{T} \cap \set{S}}} \\[0.1cm]
&& \hspace*{-0.7cm} = \Ent{X_{\set{T}}} + \Ent{X_{\set{S}}} -\MInfoCond{X_{\set{T} \setminus \set{S}}}{X_{\set{S} \setminus \set{T}}}{X_{\set{T} \cap \set{S}}} \\[0.1cm]
&& \hspace*{-0.7cm} = f(\set{T}) + f(\set{S}) -\MInfoCond{X_{\set{T} \setminus \set{S}}}{X_{\set{S} \setminus \set{T}}}{X_{\set{T} \cap \set{S}}},
\end{eqnarray}
\end{subequations}
which gives
\begin{eqnarray}  \label{eq2: submodularity - entropy}
\hspace*{-0.8cm} f(\set{T}) + f(\set{S}) - \bigl[ f(\set{T} \cup \set{S}) + f(\set{T} \cap \set{S}) \bigr]
= \MInfoCond{X_{\set{T} \setminus \set{S}}}{X_{\set{S} \setminus \set{T}}}{X_{\set{T} \cap \set{S}}} \geq 0.
\end{eqnarray}
This proves the submodularity of $f$, while also showing that
\begin{eqnarray}  \label{eq3: submodularity - entropy}
\hspace*{-0.8cm} f(\set{T}) + f(\set{S}) = f(\set{T} \cup \set{S}) + f(\set{T} \cap \set{S}) \quad \Longleftrightarrow \quad
\CondInd{X_{\set{T} \setminus \set{S}}}{X_{\set{S} \setminus \set{T}}}{X_{\set{T} \cap \set{S}}} \, ,
\end{eqnarray}
i.e., the rightmost side of \eqref{eq2: submodularity - entropy} holds with equality if and only if
$X_{\set{T} \setminus \set{S}}$ and $X_{\set{S} \setminus \set{T}}$ are conditionally independent
given $X_{\set{T} \cap \set{S}}$.
\item Monotonicity: If $\set{S} \subseteq \set{T} \subseteq \Omega$, then

\vspace*{-0.3cm}
\begin{subequations}
\begin{align}  \label{eq: monotonicity - entropy}
& \hspace*{-0.7cm} f(\set{S}) = \Ent{X_{\set{S}}} \\
&\leq \Ent{X_{\set{S}}} + \EntCond{X_{\set{T}}}{X_{\set{S}}} \\
&= \Ent{X_{\set{T}}} \\
&= f(\set{T}),
\end{align}
so $f$ is monotonically increasing.
\end{subequations}
\end{itemize}

We next prove Item~(b). Consider the set function $f$ in \eqref{set function 2}.
\begin{itemize}
\item $f(\es) = 0$, and $f(\Omega) = \Ent{X_{\Omega}}$.
\item Supermodularity: If $\set{S}, \set{T} \subseteq \Omega$, then

\vspace*{-0.3cm}
\begin{subequations}
\begin{align}
\label{eq1: supermodularity - set function 2}
\hspace*{-0.5cm} f(\set{T} \cup \set{S}) + f(\set{T} \cap \set{S})
&= \EntCond{X_{\set{T} \cup \set{S}}}{X_{\cset{T} \cap \cset{S}}}
+ \EntCond{X_{\set{T} \cap \set{S}}}{X_{\cset{T} \cup \cset{S}}} \\[0.1cm]
\label{eq2: supermodularity - set function 2}
&= \bigl[ \Ent{X_\Omega} - \Ent{X_{\cset{T} \cap \cset{S}}} \bigr]
+ \bigl[ \Ent{X_\Omega} - \Ent{X_{\cset{T} \cup \cset{S}}} \bigr] \\[0.1cm]
\label{eq3: supermodularity - set function 2}
&= 2 \Ent{X_\Omega} - \bigl[ \Ent{X_{\cset{T} \cup \cset{S}}} + \Ent{X_{\cset{T} \cap \cset{S}}}  \bigr] \\[0.1cm]
\label{eq4: supermodularity - set function 2}
&\geq 2 \Ent{X_\Omega} - \bigl[ \Ent{X_{\cset{T}}} + \Ent{X_{\cset{S}}} \bigr] \\[0.1cm]
\label{eq5: supermodularity - set function 2}
&= \bigl[ \Ent{X_\Omega} - \Ent{X_{\cset{T}}} \bigr] + \bigl[ \Ent{X_\Omega} - \Ent{X_{\cset{S}}} \bigr] \\[0.1cm]
\label{eq6: supermodularity - set function 2}
&= \EntCond{X_{\set{T}}}{X_{\cset{T}}} + \EntCond{X_{\set{S}}}{X_{\cset{S}}} \\[0.1cm]
\label{eq7: supermodularity - set function 2}
&= f(\set{T}) + f(\set{S}),
\end{align}
\end{subequations}
where inequality \eqref{eq4: supermodularity - set function 2} holds since
the entropy function in \eqref{entropic function} is submodular (by Item~(a)).
\item Monotonicity: If $\set{S} \subseteq \set{T} \subseteq \Omega$, then

\vspace*{-0.3cm}
\begin{subequations}  \label{eq: monotonicity - set function 2}
\begin{align}  \label{eq1: monotonicity - set function 2}
& \hspace*{-0.7cm} f(\set{S}) = \EntCond{X_{\set{S}}}{X_{\cset{S}}} \\[0.1cm]
\label{eq2: monotonicity - set function 2}
&\leq \EntCond{X_{\set{S}}}{X_{\cset{T}}} \quad (\cset{T} \subseteq \cset{S}) \\[0.1cm]
\label{eq3: monotonicity - set function 2}
&\leq \EntCond{X_{\set{T}}}{X_{\cset{T}}} \\[0.1cm]
\label{eq4: monotonicity - set function 2}
&= f(\set{T}),
\end{align}
\end{subequations}
so $f$ is monotonically increasing.
\end{itemize}

Item~(c) follows easily from Items~(a) and~(b). Consider the set function $f \colon 2^\Omega \to \Reals$
in \eqref{set function 3}. Then, for all $\set{T} \in \Omega$,
$f(\set{T}) = \MInfo{X_{\set{T}}}{X_{\cset{T}}} = \Ent{X_{\set{T}}} - \EntCond{X_{\set{T}}}{X_{\cset{T}}}$,
so $f$ is expressed as a difference of a submodular function and a supermodular function, which
gives a submodular function. Furthermore, $f(\es) = 0$; by the symmetry of the mutual information,
$f(\set{T}) = f(\cset{T})$ for all $\set{T} \subseteq \Omega$, so $f$ is not monotonic.

We next prove Item~(d). Consider the set function $f \colon 2^\set{V} \to \Reals$
in \eqref{set function 4}, and we need to prove that $f$ is submodular under
the conditions in Item~(d) where $\set{U}, \set{V} \subseteq \Omega$ are disjoint
subsets, and the entries of the random vector $X_{\set{V}}$ are
conditionally independent given $X_{\set{U}}$.
\begin{itemize}
\item $f(\es) = \MInfo{X_{\set{U}}}{X_{\es}} = 0$.
\item Submodularity: If $\set{S}, \set{T} \subseteq \set{V}$, then

\vspace*{-0.3cm}
\begin{subequations}  \label{eq1: submodularity - set function 4}
\begin{align}
& \hspace*{-0.3cm} f(\set{T} \cup \set{S}) + f(\set{T} \cap \set{S}) \nonumber \\
\label{eq1a: submodularity - set function 4}
&= \MInfo{X_{\set{U}}}{X_{\set{T} \cup \set{S}}} + \MInfo{X_{\set{U}}}{X_{\set{T} \cap \set{S}}} \\[0.1cm]
\label{eq1b: submodularity - set function 4}
&= \bigl[ \Ent{X_{\set{T} \cup \set{S}}} - \EntCond{X_{\set{T} \cup \set{S}}}{X_{\set{U}}} \bigr]
+ \bigl[ \Ent{X_{\set{T} \cap \set{S}}} - \EntCond{X_{\set{T} \cap \set{S}}}{X_{\set{U}}} \bigr] \\[0.1cm]
\label{eq1c: submodularity - set function 4}
&= \bigl[ \Ent{X_{\set{T} \cup \set{S}}} + \Ent{X_{\set{T} \cap \set{S}}} \bigr]
- \bigl[ \EntCond{X_{\set{T} \cup \set{S}}}{X_{\set{U}}} + \EntCond{X_{\set{T} \cap \set{S}}}{X_{\set{U}}} \bigr] \\[0.1cm]
&= \bigl[ \Ent{X_{\set{T}}} + \Ent{X_{\set{S}}}
-\MInfoCond{X_{\set{T} \setminus \set{S}}}{X_{\set{S} \setminus \set{T}}}{X_{\set{T} \cap \set{S}}} \bigr] \nonumber \\[0.1cm]
\label{eq1d: submodularity - set function 4}
& \hspace*{0.4cm} - \bigl[ \EntCond{X_{\set{T} \cup \set{S}}}{X_{\set{U}}} + \EntCond{X_{\set{T} \cap \set{S}}}{X_{\set{U}}} \bigr],
\end{align}
\end{subequations}
where equality \eqref{eq1d: submodularity - set function 4} holds by the proof of Item~(a)
(see \eqref{eq2: submodularity - entropy}). By the assumption on the conditional independence of the random
variables $\{X_v\}_{v \in \set{V}}$ given $X_{\set{U}}$, we get

\vspace*{-0.3cm}
\begin{subequations}  \label{eq2: submodularity - set function 4}
\begin{align}
\EntCond{X_{\set{T} \cup \set{S}}}{X_{\set{U}}} + \EntCond{X_{\set{T} \cap \set{S}}}{X_{\set{U}}}
\label{eq2a: submodularity - set function 4}
&= \sum_{\omega \in \set{T} \cup \set{S}} \EntCond{X_\omega}{X_{\set{U}}}
+ \sum_{\omega \in \set{T} \cap \set{S}} \EntCond{X_\omega}{X_{\set{U}}} \\[0.1cm]
\label{eq2b: submodularity - set function 4}
&= \sum_{\omega \in \set{T}} \EntCond{X_\omega}{X_{\set{U}}}
+ \sum_{\omega \in \set{S}} \EntCond{X_\omega}{X_{\set{U}}} \\[0.1cm]
\label{eq2c: submodularity - set function 4}
&= \EntCond{X_{\set{T}}}{X_{\set{U}}} + \EntCond{X_{\set{S}}}{X_{\set{U}}}.
\end{align}
\end{subequations}
Consequently, combining \eqref{eq1: submodularity - set function 4} and \eqref{eq2: submodularity - set function 4} gives

\vspace*{-0.3cm}
\begin{subequations}  \label{eq3: submodularity - set function 4}
\begin{align}
\hspace*{-0.5cm} f(\set{T} \cup \set{S}) + f(\set{T} \cap \set{S})
&= \bigl[ \Ent{X_{\set{T}}} + \Ent{X_{\set{S}}}
-\MInfoCond{X_{\set{T} \setminus \set{S}}}{X_{\set{S} \setminus \set{T}}}{X_{\set{T} \cap \set{S}}} \bigr] \nonumber \\[0.1cm]
\label{eq3a: submodularity - set function 4}
& \hspace*{0.4cm} - \bigl[ \EntCond{X_{\set{T}}}{X_{\set{U}}} + \EntCond{X_{\set{S}}}{X_{\set{U}}} \bigr] \\[0.1cm]
\hspace*{-0.5cm} &= \bigl[ \Ent{X_{\set{T}}} - \EntCond{X_{\set{T}}}{X_{\set{U}}} \bigr]
+ \bigl[ \Ent{X_{\set{S}}} - \EntCond{X_{\set{S}}}{X_{\set{U}}} \bigr]  \nonumber \\[0.1cm]
\label{eq3b: submodularity - set function 4}
& \hspace*{0.4cm} -\MInfoCond{X_{\set{T} \setminus \set{S}}}{X_{\set{S} \setminus \set{T}}}{X_{\set{T} \cap \set{S}}} \\[0.1cm]
\label{eq3c: submodularity - set function 4}
&= \MInfo{X_{\set{T}}}{X_{\set{U}}} + \MInfo{X_{\set{S}}}{X_{\set{U}}}
-\MInfoCond{X_{\set{T} \setminus \set{S}}}{X_{\set{S} \setminus \set{T}}}{X_{\set{T} \cap \set{S}}} \\[0.1cm]
\label{eq3d: submodularity - set function 4}
&= f(\set{T}) + f(\set{S}) - \MInfoCond{X_{\set{T} \setminus \set{S}}}{X_{\set{S}
\setminus \set{T}}}{X_{\set{T} \cap \set{S}}} \\[0.1cm]
\label{eq3e: submodularity - set function 4}
& \leq f(\set{T}) + f(\set{S}),
\end{align}
\end{subequations}
where the inequality \eqref{eq3e: submodularity - set function 4} holds with equality if and only if
$X_{\set{T} \setminus \set{S}}$ and $X_{\set{S} \setminus \set{T}}$ are conditionally independent
given $X_{\set{T} \cap \set{S}}$.
\item Monotonicity: If $\set{S} \subseteq \set{T} \subseteq \set{V}$, then
\begin{eqnarray}
\label{eq3f: monotonicity - set function 4}
f(\set{S}) = \MInfo{X_{\set{U}}}{X_{\set{S}}} \leq \MInfo{X_{\set{U}}}{X_{\set{T}}} = f(\set{T}),
\end{eqnarray}
so $f$ is monotonically increasing.
\end{itemize}

We finally prove Item~(e), where it is needed to show that the entropy of a sum of independent random
variables is a rank function. Let $f \colon 2^{\Omega} \to \Reals$ be the set function as given in
\eqref{set function 5}.
\begin{itemize}
\item $f(\es) = 0$.
\item Submodularity: Let $\set{S}, \set{T} \subseteq \Omega$. Define
\begin{eqnarray} \label{eqdef: U,V,W}
U \eqdef \sum_{\omega \in \set{T} \cap \set{S}} X_\omega, \quad V \eqdef \sum_{\omega \in \set{S} \setminus \set{T}} X_\omega,
\quad W \eqdef \sum_{\omega \in \set{T} \setminus \set{S}} X_\omega.
\end{eqnarray}
From the independence of the random variables $\{X_\omega\}_{\omega \in \Omega}$, it follows that $U,V$ and $W$ are independent.
Hence, we get

\vspace*{-0.3cm}
\begin{subequations}   \label{eq1: submodularity - set function 5}
\begin{align}
& \hspace*{-0.3cm} \bigl[ f(\set{T}) + f(\set{S}) \bigr] - \bigl[ f(\set{T} \cup \set{S}) + f(\set{T} \cap \set{S}) \bigr] \nonumber \\[0.1cm]
\label{eq1a: submodularity - set function 5}
& = \bigl[ f(\set{T}) - f(\set{T} \cap \set{S}) \bigr] - \bigl[ f(\set{T} \cup \set{S}) - f(\set{S}) \bigr] \\[0.1cm]
\label{eq1b: submodularity - set function 5}
& = \bigl[ \Ent{U+W} - \Ent{U} \bigr] - \bigl[ \Ent{U+V+W} - \Ent{U+V} \bigr] \\[0.1cm]
\label{eq1c: submodularity - set function 5}
& = \bigl[ \Ent{U+W} - \EntCond{U+W}{W} \bigr] - \bigl[ \Ent{U+V+W} - \Ent{U+V} \bigr] \\[0.1cm]
\label{eq1d: submodularity - set function 5}
& = \bigl[ \Ent{U+W} - \EntCond{U+W}{W} \bigr] - \bigl[ \Ent{U+V+W} - \EntCond{U+V+W}{W} \bigr] \\[0.1cm]
\label{eq1e: submodularity - set function 5}
& = \MInfo{U+W}{W} - \MInfo{U+V+W}{W} \\[0.1cm]
\label{eq1f: submodularity - set function 5}
& \geq \MInfo{U+W}{W} - \MInfo{U+W,V}{W},
\end{align}
\end{subequations}
and

\vspace*{-0.3cm}
\begin{subequations}   \label{eq2: submodularity - set function 5}
\begin{align}
\label{eq2a: submodularity - set function 5}
\MInfo{U+W,V}{W} &= \MInfo{V}{W} + \MInfoCond{U+W}{W}{V} \\
\label{eq2b: submodularity - set function 5}
&= \MInfoCond{U+W}{W}{V}  \\
\label{eq2c: submodularity - set function 5}
&= \MInfo{U+W}{W}.
\end{align}
\end{subequations}
Combining \eqref{eq1: submodularity - set function 5} and \eqref{eq2: submodularity - set function 5}
gives \eqref{eq: submodularity}.
\item Monotonicity: If $\set{S} \subseteq \set{T} \subseteq \Omega$, then since
$\{X_\omega\}_{\omega \in \Omega}$ are independent random variables, \eqref{eqdef: U,V,W}
implies that $U$ and $W$ are independent and $V=0$. Hence,

\vspace*{-0.3cm}
\begin{subequations}  \label{eq1: monotonicity - set function 5}
\begin{align}
\label{eq1a: monotonicity - set function 5}
f(\set{T}) - f(\set{S}) &= \Ent{U+W} - \Ent{U} \\
\label{eq1b: monotonicity - set function 5}
&= \Ent{U+W} - \EntCond{U+W}{W} \\
\label{eq1c: monotonicity - set function 5}
&= \MInfo{U+W}{W} \geq 0.
\end{align}
\end{subequations}
\end{itemize}
This completes the proof of Proposition~\ref{proposition: information measures, polymatroids}.

\section{Proof of Proposition~\ref{proposition: Shearer - submodularity}}
\label{Appendix: Proof - Shearer - submodularity}

\setcounter{equation}{0}
\setcounter{lemma}{0}
\renewcommand{\thelemma}{\thesection.\arabic{lemma}}
\setcounter{remark}{0}
\renewcommand{\theremark}{\thesection.\arabic{remark}}

\begin{lemma}
\label{lemma: chain}
Let $\{\set{B}_j\}_{j=1}^{\ell}$ (with $\ell \geq 2$) be a sequence of sets
that is not a chain (i.e., there is no permutation $\pi \colon \OneTo{\ell}
\to \OneTo{\ell}$ such that $\set{B}_{\pi(1)} \subseteq \set{B}_{\pi(2)}
\subseteq \ldots \subseteq \set{B}_{\pi(\ell)}$). Consider a recursive process
where, at each step, a pair of sets that are not related by inclusion is replaced
with their intersection and union. Then, there exists such a recursive process
that leads to a chain in a finite number of steps.
\end{lemma}
\begin{proof}
The lemma is proved by mathematical induction on $\ell$. It holds for $\ell=2$
since $\set{B}_1 \cap \set{B}_2 \subseteq \set{B}_1 \cup \set{B}_2$, and
the process halts in a single step.
Suppose that the lemma holds with a fixed $\ell \geq 2$, and for an
arbitrary sequence of $\ell$ sets which is not a chain.
We aim to show that it also holds for every sequence of $\ell+1$ sets which is not a chain.
Let $\{\set{B}_j\}_{j=1}^{\ell+1}$ be such an arbitrary sequence of sets, and
consider the subsequence of the first $\ell$ sets $\set{B}_1, \ldots, \set{B}_\ell$.
If it is not a chain, then (by the induction hypothesis) there exists a
recursive process as above which enables to transform it into a chain in a finite number of steps, i.e.,
we get a chain $\set{B}'_1 \subseteq \set{B}'_2 \subseteq \ldots \subseteq \set{B}'_\ell$.
If $\set{B}'_{\ell} \subseteq \set{B}_{\ell+1}$ or $\set{B}_{\ell+1} \subseteq \set{B}'_1$, then we get a chain
of $\ell+1$ sets. Otherwise, by proceeding with the recursive process where
$\set{B}'_\ell$ and $\set{B}_{\ell+1}$ are replaced with their intersection and union,
consider the sequence
\begin{eqnarray} \label{sequence 1}
\set{B}'_1, \, \ldots, \, \set{B}'_{\ell-1}, \, \set{B}'_\ell \cap \set{B}_{\ell+1},
\, \set{B}'_\ell \cup \set{B}_{\ell+1}.
\end{eqnarray}
By the induction hypothesis, the first $\ell$ sets in this sequence
can be transformed into a chain (in a finite number of steps) by a
recursive process as above; this gives a chain of the form
$ \set{B}''_1 \subseteq \set{B}''_2 \ldots \subseteq \set{B}''_{\ell-1}
\subseteq \set{B}''_\ell$.
The first $\ell$ sets in \eqref{sequence 1} are all included in $\set{B}'_{\ell}$, so
every combination of unions and intersections of these $\ell$ sets is also included in
$\set{B}'_{\ell}$. Hence, the considered recursive process leads to a chain of the form
\begin{eqnarray} \label{sequence 2}
\set{B}''_1 \subseteq \set{B}''_2 \ldots \subseteq \set{B}''_{\ell-1}
\subseteq \set{B}''_\ell \subseteq \set{B}'_\ell \cup \set{B}_{\ell+1},
\end{eqnarray}
where the last inclusion in \eqref{sequence 2} holds since $\set{B}''_\ell \subseteq \set{B}'_\ell$.
The claim thus holds for $\ell+1$ if it holds for a given $\ell$, and it holds for $\ell=2$, it therefore
holds by mathematical induction for all integers $\ell \geq 2$.
\end{proof}

We first prove Proposition~\ref{proposition: Shearer - submodularity}a.
Suppose that there is a permutation $\pi \colon \OneTo{M} \to \OneTo{M}$ such that
$\set{S}_{\pi(1)} \subseteq \set{S}_{\pi(2)} \subseteq \ldots \subseteq \set{S}_{\pi(M)}$
is a chain.
Since every element in $\Omega$ is included in at least $d$ of these subsets, then
it should be included in (at least) the $d$ largest sets of this chain, so $\set{S}_{\pi(j)} = \Omega$
for every $j \in \FromTo{M-d+1}{M}$. Due to the non-negativity
of $f$, it follows that
\begin{subequations} \label{eq: for a chain}
\begin{align}
\sum_{j=1}^M f(\set{S}_j) &\geq \sum_{j=M-d+1}^M f(\set{S}_{\pi(j)}) \\
&= d \, f(\Omega).
\end{align}
\end{subequations}
Otherwise, if we cannot get a chain by possibly permuting the subsets in the sequence
$\bigl\{\set{S}_j\}_{j=1}^M$, consider a pair of subsets $\set{S}_n$ and $\set{S}_m$
that are not related by inclusion, and replace them with their intersection and union.
By the submodularity of $f$,

\vspace*{-0.3cm}
\begin{subequations} \label{eq: reduced sums at each step}
\begin{align}
\label{eq1: reduced sums at each step}
\sum_{j=1}^M f(\set{S}_j) &= \sum_{j \neq n,m} f(\set{S}_j) + f(\set{S}_n) + f(\set{S}_m) \\
\label{eq2: reduced sums at each step}
&\geq \sum_{j \neq n,m} f(\set{S}_j) + f(\set{S}_n \cap \set{S}_m) + f(\set{S}_n \cup \set{S}_m).
\end{align}
\end{subequations}
For all $\omega \in \Omega$, let $\deg(\omega)$ be the number of indices $j \in \OneTo{M}$
such that $\omega \in \set{S}_j$. By replacing $\set{S}_n$ and $\set{S}_m$ with
$\set{S}_n \cap \set{S}_m$ and $\set{S}_n \cup \set{S}_m$, the set of values $\{\deg(\omega)\}_{\omega \in \Omega}$
stays unaffected (indeed, if $\omega \in \set{S}_n$ and $\omega \in \set{S}_m$, then it belongs to
their intersection and union; if $\omega$ belongs to only one of the sets $\set{S}_n$ and $\set{S}_m$,
then $\omega \notin \set{S}_n \cap \set{S}_m$ and $\omega \in \set{S}_n \cup \set{S}_m$; finally, if
$\omega \notin \set{S}_n$ and $\omega \notin \set{S}_m$, then it does not belong to their intersection
and union). Now, consider the recursive process in Lemma~\ref{lemma: chain}. Since the profile of the
number of inclusions of the elements in $\Omega$ is preserved in each step of the recursive process
in Lemma~\ref{lemma: chain}, it follows that every element in $\Omega$ stays to belong to at least
$d$ sets in the chain which is obtained at the end of this recursive process.
Moreover, in light of \eqref{eq: reduced sums at each step}, in every step
of the recursive process in Lemma~\ref{lemma: chain}, the sum in the LHS of
\eqref{eq: reduced sums at each step} cannot increase. Inequality \eqref{06.03.2021b1} therefore
finally follows from the earlier part of the proof for a chain (see \eqref{eq: for a chain}).

We next prove Proposition~\ref{proposition: Shearer - submodularity}b.
Let $\set{A} \subset \Omega$, and suppose that every element in $\set{A}$ is included in at
least $d \geq 1$ of the subsets $\{\set{S}_j\}_{j=1}^M$. For all $j \in \OneTo{M}$, define
$\set{S}'_j \eqdef \set{S}_j \cap \set{A}$, and consider the sequence
$\bigl\{\set{S}'_j \bigr\}_{j=1}^M$ of subsets of $\set{A}$.
If $f$ is a rank function, then it is monotonically increasing, which yields
\begin{eqnarray} \label{06.03.2021a1}
f(\set{S}'_j) \leq f(\set{S}_j), \qquad j \in \OneTo{M}.
\end{eqnarray}
Each element of $\set{A}$ is also included in at least $d$ of the subsets
$\bigl\{\set{S}'_j \bigr\}_{j=1}^M$ (by construction, and since (by assumption) each element
in $\set{A}$ is included in at least $d$ of the subsets $\{\set{S}_j\}_{j=1}^M$).
By the non-negativity and submodularity of $f$,
Proposition~\ref{proposition: Shearer - submodularity}a gives
\begin{eqnarray}  \label{06.03.2021a2}
\sum_{j=1}^M f(\set{S}'_j) \geq d \, f(\set{A}).
\end{eqnarray}
Combining \eqref{06.03.2021a1} and \eqref{06.03.2021a2} yields \eqref{06.03.2021b2}.
This completes the proof of Proposition~\ref{proposition: Shearer - submodularity}.

\begin{remark}
Lemma~\ref{lemma: chain} is weaker than a claim that, in every recursive process as
in Lemma~\ref{lemma: chain}, the number of pairs of sets that are not related by inclusion
is strictly decreasing at each step. Lemma~\ref{lemma: chain} is, however, sufficient for
our proof of Proposition~\ref{proposition: Shearer - submodularity}a.
\end{remark}

\end{document}